\documentclass[11pt,a4paper]{article}
\usepackage[utf8]{inputenc}
\usepackage[margin=1in]{geometry}
\usepackage{color}
\usepackage{graphicx}
\usepackage{microtype}
\usepackage{array}
\usepackage{verbatim}
\usepackage{caption}
\usepackage{subcaption}
\usepackage{amsmath,amsthm,amsfonts,amssymb,latexsym}
\usepackage{bbm}
\usepackage{setspace}
\usepackage{xparse}
\usepackage{epstopdf}
\usepackage{pgf}
\usepackage[colorlinks=true,citecolor=blue]{hyperref}
\usepackage[nameinlink,capitalise]{cleveref}

\usepackage[style=trad-abbrv,doi=false,url=false,isbn=false,backend=biber]{biblatex}
\DeclareFieldFormat{volume}{volume \textbf{#1}}
\DeclareFieldFormat[article]{volume}{\textbf{#1}}

\usepackage{tikz}
\usepackage{tikz-cd}
\usepackage{pgfplotstable}
\pgfplotsset{compat=1.14}
\usetikzlibrary{patterns}
\usetikzlibrary{calc}
\usetikzlibrary{angles}
\usetikzlibrary{quotes}
\usetikzlibrary{external}

\DeclareDocumentCommand\abs{s m} {\IfBooleanTF{#1}{\left|#2\right|}{\left|#2\right|}}
\DeclareDocumentCommand\lp{m m o} {L^{#1}\left(#2 \IfNoValueF{#3}{,#3}\right)}
\DeclareDocumentCommand\norm{s m o} {\IfBooleanTF{#1}{\|#2\|}{\left\|#2\right\|}\IfNoValueF{#3}{_{#3}}}
\DeclareDocumentCommand\seminorm{m o o} {\left|#1\right|\IfNoValueF{#2}{_{#2 \IfNoValueF{#3}{,#3}}}}
\DeclareDocumentCommand\ip{s m m o} {\IfBooleanTF{#1}{\langle #2,#3 \rangle}{\left\langle #2,#3 \right\rangle}\IfNoValueF{#4}{_{#4}}}
\DeclareDocumentCommand\bigo{s o m} {\mathcal O\IfNoValueF{#2}{_{#2}}\IfBooleanTF{#1}{(#3)}{\left(#3\right)}}

\DeclareMathOperator{\sign}{sign}

\DeclareMathOperator{\id}{id}

\newcommand{\e}{\mathrm{e}}

\newcommand{\commut}[2]{[#1, #2]}
\newcommand{\laplacian}{\Delta}

\newcommand{\dummy}{\,\cdot\,}
\newcommand{\expect}[0]{\mathbf{E}}

\newcommand{\var}[0]{\mathbf{V}}

\newcommand{\nat}{\mathbf N}
\newcommand{\poly}{\mathbf P}
\newcommand{\real}{\mathbf R}
\newcommand{\integer}{\mathbf Z}
\newcommand{\torus}{\mathbf T}
\newcommand{\grad}{\nabla}
\newcommand{\imag}{\mathrm{i}}

\newcommand{\vect}[1]{\boldsymbol{\mathbf #1}}
\newcommand{\mat}[1]{\vect #1}

\renewcommand{\d}{\mathrm d}
\renewcommand{\t}{\mathsf T}

\makeatletter
\DeclareDocumentCommand \derivative{s m o m}{%
    \def\@der{\IfBooleanTF{#1}{\mathrm{d}}{\partial}}
    \def\@default{%
        \mathchoice{%
                \frac{%
                    \@der\ifnum\pdfstrcmp{#2}{1}=0\else^{#2}\fi {\IfNoValueTF{#3}{}{#3}}
                }{%
                    \@for\@token:={#4}\do{\@der \@token}
                }
            } {%
                \@for\@token:={#4}\do{\@der_\@token} \IfNoValueTF{#3}{}{#3}
            } {} {}
    }
    \IfBooleanTF{#1}{\IfNoValueTF{#3}{\@default}{%
                #3%
                \ifnum\pdfstrcmp{#2}{1}=0'\else%
                \ifnum\pdfstrcmp{#2}{2}=0''\else%
                \ifnum\pdfstrcmp{#2}{3}=0^{(3)}\else%
                \ifnum\pdfstrcmp{#2}{4}=0^{(4)}\else%
                \ifnum\pdfstrcmp{#2}{5}=0^{(5)}\else%
                ^{(#2)}\fi\fi\fi\fi\fi
            }
        }{\@default}
}
\makeatother

\definecolor{darkred}{rgb}{.5,0,0}
\definecolor{darkgreen}{rgb}{0,.5,0}
\definecolor{darkblue}{rgb}{0,0,.5}

\theoremstyle{plain}

\newtheorem{theorem}{Theorem}[section]
\newtheorem{lemma}[theorem]{Lemma}
\newtheorem{corollary}[theorem]{Corollary}
\newtheorem{proposition}[theorem]{Proposition}

\newtheorem{remark}{Remark}[section]
\newtheorem{example}{Example}[section]
\numberwithin{equation}{section}

\newcounter{urbainCounter}

\crefname{equation}{}{}
\crefname{paragraph}{\S\!}{\S}
\crefname{figure}{Figure}{Figures}

\newcommand{\email}[1]{\href{#1}{#1}}

\newcommand{\orcid}[1]{\href{https://orcid.org/#1}{\includegraphics[width=.4cm]{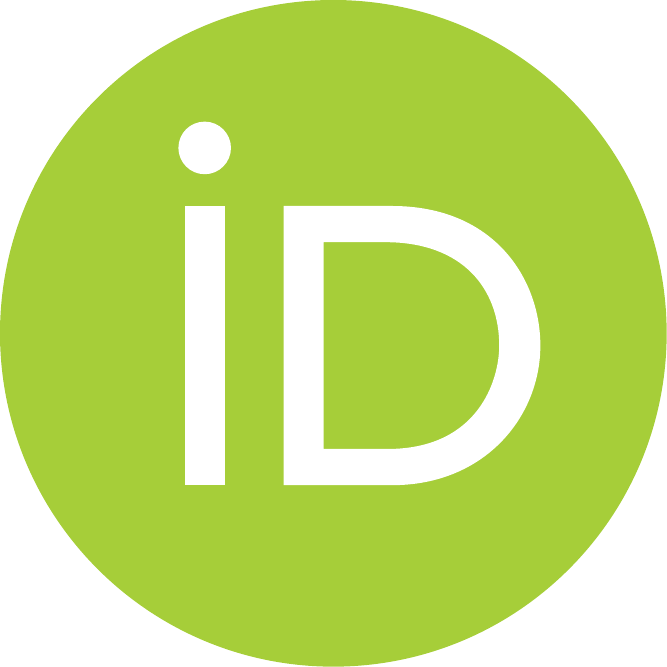}}}

\usepackage{enumerate}

\renewcommand{\leq}{\leqslant}
\renewcommand{\geq}{\geqslant}

\usepackage{mathrsfs}

\date{\today}
\title{Mobility estimation for Langevin dynamics using control variates}
\author{%
  G.A. Pavliotis\thanks{Department of Mathematics, Imperial College London (\email{g.pavliotis@imperial.ac.uk})}%
  \hspace{2mm}\orcid{0000-0002-3468-9227}%
  \and G. Stoltz\thanks{CERMICS, \'Ecole des Ponts, France \& MATHERIALS, Inria Paris (\email{gabriel.stoltz@enpc.fr})}
  \hspace{2mm}\orcid{0000-0002-2797-5938}%
  \and U. Vaes\thanks{MATHERIALS, Inria Paris (\email{urbain.vaes@inria.fr})}%
  \hspace{2mm}\orcid{0000-0002-7629-7184}
}

\begin{document}
\maketitle

\begin{abstract}
    The scaling of the mobility of two-dimensional Langevin dynamics in a periodic potential as the friction vanishes is not well understood for non-separable potentials.
    Theoretical results are lacking,
    and numerical calculation of the mobility in the underdamped regime is challenging because
    the computational cost of standard Monte Carlo methods is inversely proportional to the friction coefficient,
    while deterministic methods are ill-conditioned.
    In this work, we propose a new variance-reduction method based on control variates for efficiently estimating the mobility of Langevin-type dynamics.
    We provide bounds on the bias and variance of the proposed estimator,
    and illustrate its efficacy through numerical experiments,
    first in simple one-dimensional settings
    and then for two-dimensional Langevin dynamics.
    Our results corroborate previous numerical evidence that
    the mobility scales as~$\gamma^{-\sigma}$, with~$0 < \sigma \leq 1$,
    in the low friction regime for a simple non-separable potential.
\end{abstract}


\section{Introduction}%
\label{sec:introduction}
Langevin dynamics model the evolution of a system of particles interacting with an environment at fixed temperature.
They are widely used for the calculation of macroscopic properties of matter in molecular simulation~\cite{MR2723222,allen2017computer}.
Assuming a diagonal mass matrix,
the standard Langevin dynamics, sometimes called underdamped Langevin dynamics,
reads after appropriate non-dimensionalization~\cite[Section 2.2.4]{MR2681239}
\begin{subequations}
\label{eq:langevin}
\begin{align}
    \label{eq:langevin_q}
    \d \vect q_t &= \vect p_t \, \d t, \\
    \label{eq:langevin_p}
    \d \vect p_t &= - \grad V(\vect q_t) \, \d t - \gamma \, \vect p_t \, \d t + \sqrt{2 \gamma \beta^{-1}} \, \d \vect w_t.
\end{align}
\end{subequations}
Here, $\vect q_t \in \torus^d$ and $\vect p_t \in \real^d$ are the position and velocity variables,
with~$\torus^d = \real^d / 2\pi \integer^d$ the $d$-dimensional torus with period $2 \pi$.
Throughout this work, we emphasize vectorial quantities in bold.
The parameter $\gamma > 0$ is a dimensionless parameter called friction,
$\beta > 0$ is inversely proportional to the temperature,
$V$ is a smooth periodic potential
and~$\vect w_t$ is a standard $d$-dimensional Brownian motion.
The dynamics~\eqref{eq:langevin} is ergodic with respect to the Boltzmann--Gibbs probability measure
\begin{equation}
    \label{eq:invariant_measure}
    \mu(\d \vect q \, \d \vect p) = \frac{1}{Z} \exp \bigl( - \beta H(\vect q, \vect p)  \bigr) \, \d \vect q \, \d \vect p,
    \qquad H(\vect q, \vect p) = V(\vect q) + \frac{\abs{\vect p}^2}{2},
\end{equation}
with $Z< \infty$ the normalization constant.
It will be convenient to also introduce the marginal distributions
\begin{equation}
    \label{eq:definition_prob_measures}
    \nu(\d \vect q) = \frac{\e^{- \beta V(\vect q)} \, \d \vect q}{\int_{\torus^d}\e^{-\beta V}},
    \qquad \kappa(\d \vect p) = \left( \frac{\beta}{2 \pi} \right)^{d/2}\exp \biggl( - \beta \frac{\abs*{\vect p}^2}{2} \biggr) \d \vect p.
\end{equation}

\paragraph{Definition of the mobility.}
The mobility in the direction $\vect e \in \real^d$ (with $|\vect e| = 1$)
for the dynamics~\eqref{eq:langevin} provides information on the behavior of the system
in response to an external forcing $\eta \vect e$ with magnitude~$\eta$ on the velocity process.
By analogy with macroscopic laws,
it is defined as the proportionality constant,
in the limit of a small forcing,
between the induced average velocity and the strength  of the forcing.
More precisely,
the mobility in the direction $\vect e$ is defined mathematically as
\begin{equation}
    \label{eq:relation_mobility_diffusion}
    M^{\gamma}_{\vect e} =  \lim_{\eta \to 0} \frac{1}{\eta}\expect_{\mu_{\eta}} [\vect e^\t \vect p] ,
\end{equation}
where $\mu_{\eta}$ is the invariant probability distribution of~\eqref{eq:langevin} when
an additional drift term $\eta \vect e$ is present on the right-hand side of~\eqref{eq:langevin_p}.
Let us emphasize that this additional drift term is not the gradient of a smooth periodic potential.
Nonetheless, it is possible to show that the probability measure $\mu_{\eta}$ exists and is uniquely defined,
and that the limit in~\eqref{eq:relation_mobility_diffusion} is well-defined;
see~\cite[Section 5]{MR3509213}.
Except when $\eta = 0$, in which case we recover~\eqref{eq:invariant_measure},
the measure~$\mu_{\eta}$ is not known explicitly,
and so $M_{\vect e}^{\gamma}$ cannot be obtained simply by numerical integration of the observable $\vect e^\t \vect p$ with respect to this measure.
It is well known,
based on the seminal works of Sutherland~\cite{sutherland1905lxxv}, Einstein~\cite{einstein1905molekularkinetischen} and Smoluchowski~\cite{von1906kinetischen} in the early 1900s,
that the mobility coincides
(up to the factor $\beta$)
with the so-called effective diffusion coefficient associated with the dynamics,
which opens the door to the simple Monte Carlo approach based on~\eqref{eq:naive_estimator} below for its estimation.
This link between mobility and diffusion,
known Eisntein's relation,
is made precise in the next paragraph,
where we also define the effective diffusion coefficient precisely.
For a rigorous justification of Einstein's relation in the specific setting of the Langevin dynamics~\eqref{eq:langevin},
we refer to~\cite[Section~5.2]{MR3509213}; see also \cite[Section~3]{LMS16} and~\cite[Chapter~9]{pavliotis2011applied}.


\paragraph{Effective diffusion.}
The concept of effective diffusion,
for the Langevin dynamics~\eqref{eq:langevin},
refers to the following functional central limit theorem:
the diffusively rescaled position process $(\varepsilon \vect q_{t/\varepsilon^2})_{t\geq0}$ converges as $\varepsilon \to 0$,
weakly in the space of continuous functions over compact time intervals,
to a Brownian motion in $\real^d$ with a matrix prefactor~$\sqrt{2 \mat D^{\gamma}}$.
The matrix $\mat D^{\gamma}$ is known as the effective diffusion matrix associated with the dynamics.
This result may be obtained by using the homogenization technique pioneered by Bhattacharya in~\cite{MR663900},
which hinges on the functional central limit theorem for martingales~\cite{MR668684};
see also~\cite[Chapter~3]{MR503330} for early results concerning the asymptotic analysis of SDEs,
the book~\cite[Chapter 18]{pavliotis2008multiscale} for a pedagogical presentation of homogenization for stochastic differential equations,
and~\cite[Theorem 2.5]{MR2793823} for a detailed proof of the homogenization theorem for the Markovian approximation of the generalized Langevin equation.
The precise statement of Einstein's relation is then that
\[
    D^{\gamma}_{\vect e} := \vect e^\t \mat D^{\gamma} \vect e = \beta M^{\gamma}_{\vect e}.
\]

\paragraph{Link with the Poisson equation.}
The effective diffusion coefficient
can be expressed in terms of the solution to a partial differential equation (PDE) involving the generator of the Markov semigroup associated with~\eqref{eq:langevin},
which is given by
\begin{equation}
    \label{eq:decomposition_generator}
    \mathcal L
    = \vect p \cdot \grad_{\vect q} - \grad V \cdot \grad_{\vect p} + \gamma \left( - \vect p \cdot \grad_{\vect p} + \beta^{-1} \laplacian_{\vect p} \right)
    =: \mathcal L_{\rm Ham} + \gamma \mathcal L_{\rm FD}.
\end{equation}
Specifically, it is possible to show~\cite{MR663900} that
\begin{equation}
    \label{eq:effective_diffusion_poisson}
    D^{\gamma}_{\vect e} = \ip{\phi_{\vect e}}{\vect e^\t p},
\end{equation}
where $\phi_{\vect e}$ denotes the unique solution to the Poisson equation
\begin{equation}
    \label{eq:poisson_equation}
    - \mathcal L \phi_{\vect e} = \vect e^\t \vect p,
    \qquad \phi_{\vect e} \in L^2_0(\mu) := \bigl\{ u \in L^2(\mu): \ip{u}{1} = 0 \bigr\}.
\end{equation}
Throughout this work,
$\ip{\dummy}{\dummy}$ and $\norm{\dummy}$ denote respectively the inner product and norm of $\lp{2}{\mu}$
unless otherwise specified.
Several techniques can be employed in order to show that~\eqref{eq:poisson_equation} admits a unique solution in $L^2_0(\mu)$
for any right-hand side in $L^2_0(\mu)$:
one may use the approach employed in~\cite[Proposition 5.1]{MR2793823},
which is itself inspired from~\cite[Lemma 2.1]{MR812349},
or obtain well-posedness as a corollary of the exponential decay in $L^2(\mu)$ of the Markov semigroup associated with the dynamics,
as in~\cite[Corollary 1]{roussel2018spectral}.
See also~\cite{Herau06,MR3106879,MR3522857,BFLS20} for other references on the exponential decay for semigroups with a hypocoercive generator.

\paragraph{Numerical estimation of the mobility.}
In spatial dimension 1,
it is possible to obtain an accurate estimation of the effective diffusion coefficient by solving the Poisson equation~\eqref{eq:poisson_equation} using a deterministic method~\cite{roussel2018spectral},
but this approach is generally too computationally expensive in higher dimensions.
In spatial dimension~2, for example,
a spectral discretization of~\eqref{eq:poisson_equation} based on a tensorized basis of functions,
with say~$N$ degrees of freedom per dimension of the state space $\torus^2 \times \real^2$, leads to a linear system with~$N^4$ unknowns,
which is computationally intractable for large values of $N$.
In this setting, probabilistic methods offer an attractive alternative.
It follows from the definition of~$D^{\gamma}_{\vect e}$ that,
for any $t > 0$,%
\begin{equation}
    \label{eq:einsteins_formula}
    D^{\gamma}_{\vect e}
    =\lim_{\varepsilon \to 0} \frac{\expect\Bigl[\bigl\lvert \vect e^\t \left(\varepsilon \vect q_{t/\varepsilon^2} - \varepsilon \vect q_0\right) \bigr\rvert^2\Bigr]}{2 t}
    =\lim_{T \to \infty} \frac{\expect \Bigl[\bigl\lvert \vect e^\t \left(\vect q_T - \vect q_0\right) \bigr\rvert^2\Bigr]}{2T},
\end{equation}
suggesting that this coefficient may be calculated by
estimating the mean square displacement at a sufficiently large time of the equilibrium dynamics~\eqref{eq:langevin}
using Monte Carlo simulation,
which is one of the approaches taken in~\cite{MR2427108}.
Specifically, given a number $J$ of realizations of the dynamics~\eqref{eq:langevin} over a sufficiently long time interval $[0, T]$,
the effective diffusion coefficient in direction $\vect e$ may be estimated as
\begin{equation}
    \label{eq:naive_estimator}
    \widehat D^{\gamma}_{\vect e}
    = \frac{1}{J} \sum_{j=1}^{J} \frac{\left\lvert \vect e^\t \left(\vect q^{(j)}_T - \vect q^{(j)}_0\right) \right\rvert^2}{2T},
\end{equation}
where $(\vect q_t^{(j)}, \vect p_t^{(j)})_{t \geq 0}$, for $1 \leq j \leq J$,
are independent realizations of the solution to the Langevin equation~\eqref{eq:langevin} starting from i.i.d.\ initial conditions $\bigl(\vect q^{(j)}_0, \vect p^{(j)}_0\bigr)$.
The variance reduction approach we propose in the next section aims at reducing the mean square error of estimators of this type.

Another possible approach for estimating the mobility is to rely on a numerical approximation of the Green--Kubo formula;
see \cite[Section 5.1.3]{MR3509213} for general background information on this subject.
The bias associated with this approach is studied carefully in~\cite{LMS16},
and bounds on the variance are obtained in~\cite{PSW21},
showing that the variance increases linearly with the integration time over which correlations are computed.
In practice, choosing the integration time is a delicate task:
it needs to be sufficiently large to ensure that the systematic bias is small,
but not too large, or else the variance of the resulting estimator is large.
The technical challenges impeding adoption of the Green--Kubo formalism,
as well as some solutions to overcome these in the context of heat transport,
are discussed in~\cite{ercole2017accurate,Baroni2020}.

\paragraph{Overdamped and underdamped limits.}
The behavior of the Langevin dynamics~\eqref{eq:langevin} depends on the value of the friction parameter~$\gamma$.
The overdamped limit $\gamma \to \infty$ is well understood;
in this limit, the rescaled position process $(\vect q_{\gamma t})_{t \geq 0}$
converges, weakly in the space of continuous functions~\cite{MR4054345}
and almost surely uniformly over compact subintervals of $[0, \infty)$~\cite[Theorem 10.1]{MR0214150},
to the solution of the overdamped Langevin equation
\begin{equation}
    \label{eq:overdamped_langevin}
    \d \vect q_t = - \grad V(\vect q_t) \, \d t + \sqrt{2 \beta^{-1}} \, \d \vect b_t,
\end{equation}
where $\vect b_t$ is another standard $d$-dimensional Brownian motion.
It is also possible to prove that~$\gamma \mat D^{\gamma} = \mat D^{\rm ovd} + \bigo{\gamma^{-2}}$ as $\gamma \to \infty$,
where~$\mat D^{\rm ovd}$ is the effective diffusion coefficient of overdamped Langevin dynamics,
and to derive explicit expressions for the correction terms by asymptotic analysis~\cite{MR2394704}.
The diffusion coefficient in the overdamped limit is given by $\vect e^\t \mat D^{\rm ovd} \vect e = \norm{\vect e + \grad \chi_{\vect e}}[L^2(\nu)]^2$,
where $\chi_{\vect e}$ is the unique solution in~$L^2_0(\nu)$ to the Poisson equation
\[
    - \mathcal L_{\rm ovd} \chi_{\vect e} = - \vect e^\t \grad V(q), \qquad \mathcal L_{\rm ovd} = - \grad V \cdot \grad + \beta^{-1} \laplacian,
\]
with $\mathcal L_{\rm ovd}$ is the generator of the Markov semigroup associated with~\eqref{eq:overdamped_langevin}.
The reasoning in~\cite[Proposition 4.1]{MR2394704},
when appropriately generalized to the multi-dimensional setting,
shows that~$D^{\rm ovd}_{\vect e}$ is in fact an upper bound for $\gamma D^{\gamma}_{\vect e}$ for all $\gamma > 0$.

The underdamped limit is much more difficult to analyze,
especially in the multi-dimensional setting.
In spatial dimension one, it was shown in~\cite{MR2394704} that $\gamma D^{\gamma} \to D^{\rm und}$ as $\gamma \to 0$ for some limit~$D^{\rm und}$
that is also a lower bound for $\gamma D^{\gamma}$ for all $\gamma > 0$.
It is also possible~\cite[Lemma~3.4]{MR2394704}, in this case,
to show that the solution to the Poisson equation~\eqref{eq:poisson_equation},
when multiplied by $\gamma$, converges in~$\lp{2}{\mu}$ as $\gamma \to 0$ to a limit
which can be calculated explicitly in simple settings~\cite{MR2427108}.
Despite the existence of an asymptotic result,
calculating the mobility for small $\gamma$ is challenging.
Indeed, it can be shown that the spectral gap in $L^2(\mu)$ of the generator $\mathcal L$ behaves as $\mathcal O(\gamma)$ in the limit as $\gamma \to 0$~\cite{MR2394704,MR3106879,MR3522857,roussel2018spectral},
and so deterministic methods for solving the Poisson equation~\eqref{eq:poisson_equation} are ill-conditioned in this limit,
while Monte Carlo based methods are very slow to converge,
as discussed in \cref{sec:method}.

The aforementioned asymptotic result for the underdamped limit extends to the multi-dimensional setting only when the potential is separable,
that is when $V$ can be decomposed as $V(\vect q) = \sum_{i=1}^d V_i(q_i)$, corresponding to a completely integrable Hamiltonian system for $\gamma =0$, but no theoretical results exist in the non-separable case,
which was explored mostly by means of numerical experiments.
Early numerical results in~\cite{chen1996surface}, obtained from Einstein's formula~\eqref{eq:einsteins_formula},
suggest that the effective diffusion coefficient scales as~$\gamma^{-1/2}$ in the underdamped regime for a particular case of a non-separable periodic potential.
Later, in~\cite{Braun02},
different authors note that this behavior as~$\gamma^{-1/2}$ is valid only when $\gamma \in [0.01, 0.1]$,
but not for smaller values of the damping coefficient.
They conclude from simulation results that the effective diffusion coefficient scales as~$\gamma^{-\sigma}$ with $0 \leq \sigma \leq 1/3$ in the underdamped regime,
and suggest that $\sigma$ could be zero for all non-separable potentials.
More recently, in his doctoral thesis~\cite{roussel_thesis},
Roussel calculates the mobility of Langevin dynamics using a control variate approach for linear response,
relying on~\eqref{eq:relation_mobility_diffusion}.
The control variate he employs is constructed from an approximate solution to the Poisson equation~\eqref{eq:poisson_equation}.
His results suggest that, for a wide range of friction coefficients in the interval $[10^{-3}, 1]$
and in the particular case of the potential
\begin{equation}
    \label{eq:potential_julien}
    V(\vect q) = - \bigl( \cos(q_1) + \cos(q_2) \bigr) + \delta \exp \bigl(\sin(q_1 + q_2)\bigr),
\end{equation}
the mobility scales as $\gamma^{- \sigma}$,
with an exponent $\sigma \in [0, 1]$ that depends on the degree $\delta$ of non-separability of the potential. Despite claims in the physics literature, it is not expected that a universal scaling of mobility, or, equivalently, the effective diffusion coefficient, exists for general classes of non-separable potentials in dimensions higher than one.

\paragraph{Our contributions.}
In this work,
we propose a new variance reduction methodology for calculating the mobility of Langevin-type dynamics.
Like the approach in~\cite{roussel_thesis},
our methodology is based on a control variate constructed from an approximate solution to the Poisson equation~\eqref{eq:poisson_equation},
but it relies on Einstein's formula~\eqref{eq:einsteins_formula} instead of the linear response result~\eqref{eq:relation_mobility_diffusion}.
The advantages of relying on Einstein's formula are twofold:
on the one hand the associated estimators, which are based on~\eqref{eq:naive_estimator},
are asymptotically unbiased,
and on the other hand,
their calculation requires only the first derivatives of the approximate solution to the Poisson equation,
which enables to circumvent regularity issues encountered in~\cite{roussel_thesis} in the underdamped limit.


Our contributions in this work are the following.
\begin{itemize}
    \item
        We derive bounds on the bias and variance of the proposed estimator for the simple case of one-dimensional Langevin dynamics,
        in terms of the error on the solution to the Poisson equation~\eqref{eq:poisson_equation}.
        Our estimates show, in particular, that the Langevin dynamics should be integrated up to a time scaling as $\max(\gamma^{-1}, \gamma)$ in order to control the bias of the estimator.
    \item
        We examine the performance of the approach for two different approximate solutions to the Poisson equation:
        one is obtained through the Fourier/Hermite Galerkin method developed in~\cite{roussel2018spectral},
        and the other is calculated from the limiting solution of the Poisson equation in the underdamped limit;
        see~\cite{MR2427108}.
    \item
        We apply the proposed variance reduction approach to the estimation of mobility for two-dimensional Langevin dynamics in a non-separable periodic potential.
        To this end, we construct an approximation to the Poisson equation by tensorization of approximations obtained in one spatial dimension.
        We numerically study the performance of this approach,
        and present numerical results corroborating the asymptotic behavior as $\gamma^{-\sigma}$ for $\sigma \in (0, 1]$ of the effective diffusion coefficient
        observed in~\cite{roussel_thesis}.
    \item
        Using the proposed variance reduction approach
        for calculating the diffusion coefficient of generalized Langevin dynamics in the underdamped regime,
        we provide numerical evidence supporting the asymptotic behavior of the effective diffusion coefficient conjectured in our previous work~\cite{GPGSUV21} using formal asymptotics.
\end{itemize}
The rest of the paper is organized as follows.
In~\cref{sec:method},
we present a control variate approach for improving the naive Monte Carlo estimator~\eqref{eq:naive_estimator},
and obtain bounds on the bias and variance of the improved estimator in the particular case of Langevin dynamics~\eqref{eq:langevin}.
In~\cref{sec:application_to_one_dimensional_langevin_type_dynamics},
we employ the proposed approach for calculating the mobility of one-dimensional Langevin and generalized Langevin dynamics,
as a proof of concept,
and we assess the performance of various control variates in terms of variance reduction.
In~\cref{sec:applications_2d},
we present numerical results for two-dimensional Langevin dynamics,
exhibiting a scaling as $\gamma^{-\sigma}$ of the mobility in the underdamped regime.
\Cref{sec:conclusions_and_perspectives_for_future_work} is reserved for conclusions and perspectives for future work,
while the appendices contain technical results employed in~\cref{sec:application_to_one_dimensional_langevin_type_dynamics}.

\section{Improved Monte Carlo estimator for the diffusion coefficient}%
\label{sec:method}%

Throughout this section,
we focus on the Langevin dynamics~\eqref{eq:langevin} for simplicity.
Although some of our arguments are tailored specifically to this dynamics,
our approach may in principle be applied to other Langevin-type dynamics,
such as the generalized Langevin dynamics considered in \cref{sub:generalization_to_generalized_langevin_dynamics}.
We assume throughout the section that $(\vect q_t, \vect p_t)_{t\geq 0}$ is a solution of~\eqref{eq:langevin} with statistically stationary initial condition~$(\vect q_0, \vect p_0) \sim \mu$ independent of the Brownian motion $(\vect w_t)_{t \geq 0}$.
This is not a restrictive assumption in our setting as the probability measure $\mu$,
being defined explicitly on the low-dimensional space $\torus^d \times \real^d$,
can be sampled efficiently using standard methods,
for instance by rejection sampling.

Let us fix a direction $\vect e \in \real^d$, with $|\vect e| = 1$,
and denote again by $\phi_{\vect e}$ the corresponding solution to the Poisson equation~\eqref{eq:poisson_equation}.
Since the number of independent realizations in Monte Carlo estimators
appears only as a denominator in the variance,
we study estimators based on one realization only.
That is, instead of~\eqref{eq:naive_estimator}, we take as point of comparison the naive estimator
\begin{equation}
    \label{eq:simple_estimator}
    u(T) = \frac{\abs{\vect e^\t (\vect q_T - \vect q_0)}^2}{2T}.
\end{equation}

This section is divided into three parts.
In \cref{sub:construction_of_an_improved_estimator},
we construct a Monte Carlo estimator for the effective diffusion coefficient that improves on~\eqref{eq:simple_estimator}.
We then demonstrate in \cref{sub:bias} and \cref{sub:variance} that,
at least in certain parameter regimes,
this estimator has better properties than~\eqref{eq:simple_estimator} in term of bias and variance, respectively.

\subsection{Construction of an improved estimator}%
\label{sub:construction_of_an_improved_estimator}
In order to motivate the construction of an improved estimator for $D^{\gamma}_{\vect e}$,
we apply It\^o's formula to the solution $\phi_{\vect e}$ to the Poisson equation~\eqref{eq:poisson_equation},
which gives
\[
    \phi_{\vect e}(\vect q_T, \vect p_T) - \phi_{\vect e}(\vect q_0, \vect p_0)
    = - \int_{0}^{T} \vect e^\t \vect p_t \, \d t + \sqrt{2 \gamma \beta^{-1}} \int_{0}^{T} \grad_{\vect p} \phi_{\vect e}(\vect q_t, \vect p_t) \cdot \d \vect w_t.
\]
Rearranging the terms,
we obtain
\begin{equation}
    \label{eq:ito_for_phi}
    \vect e^\t(\vect q_T - \vect q_0) =
    \phi_{\vect e}(\vect q_0, \vect p_0) - \phi_{\vect e}(\vect q_T, \vect p_T)
    + \sqrt{2 \gamma \beta^{-1}} \int_{0}^{T} \grad_{\vect p} \phi_{\vect e}(\vect q_t, \vect p_t) \cdot \d \vect w_t.
\end{equation}
The estimator we propose requires the knowledge of an approximation $\psi_{\vect e}$ of the solution $\phi_{\vect e}$ to the Poisson equation~\eqref{eq:poisson_equation}.
Two concrete methods for obtaining such an approximation in the small $\gamma$ regime are presented in \cref{sec:application_to_one_dimensional_langevin_type_dynamics}.
In this section, we assume that such an approximation is given.
Let us introduce
\begin{align}
    \label{eq:definition_control_variate}
    \xi_T = \psi_{\vect e}(\vect q_0, \vect p_0) - \psi_{\vect e}(\vect q_T, \vect p_T)
    + \sqrt{2 \gamma \beta^{-1}} \int_{0}^{T} \grad_{\vect p} \psi_{\vect e}(\vect q_t, \vect p_t) \cdot \d \vect w_t.
\end{align}
\begin{remark}
    By It\^o's formula,
    it would have been equivalent, in the case where $\psi_{\vect e}$ is smooth,
    to define
    \(
        \xi_T = \int_{0}^{T} \mathcal L \psi_{\vect e}(\vect q_t, \vect p_t) \, \d t.
    \)
    However, the definition~\eqref{eq:definition_control_variate} makes sense even if $\psi_{\vect e}$ is differentiable only once,
    and so it is more widely applicable.
    In~\cref{sub:underdamped_approach}, for example, we construct a singular approximation~$\psi_{\vect e}$ that is not twice weakly differentiable.
\end{remark}
Since $\xi_T$ is expected to be a good approximation of $\vect e^\t(\vect q_T - \vect q_0)$,
in some appropriate sense,
when~$\psi_{\vect e}$ is a good approximation of~$\phi_{\vect e}$,
one may achieve a reduction in variance by using the former as a control variate for the latter.
More precisely, we consider the following estimator instead of~$u(T)$:
\begin{equation}
    \label{eq:simple_estimator_improvement_1}
     \frac{\bigl\lvert \vect e^\t (\vect q_T - \vect q_0) \bigr\rvert^2}{2T}  -	\alpha  \left( \frac{\bigl\lvert \xi_T \bigr\rvert^2}{2T} - \expect \left[ \frac{\lvert \xi_T \rvert^2}{2T} \right] \right)
    =: u(T) - \alpha \Bigl(\widehat u(T) - \expect \left[\widehat u(T)\right]\Bigr).
\end{equation}
Clearly, this estimator and $u(T)$ have the same expectation, and thus the same bias.
By standard properties of control variates~\cite{kroese2013handbook},
the value of $\alpha$ minimizing the variance can be expressed in terms of the variance of $u(T)$ and
the covariance between $u(T)$ and $\widehat u(T)$.
For simplicity of the analysis,
we consider only the case $\alpha = 1$,
which is the variance-minimizing choice when~$\widehat u(T) = u(T)$.
We mention in passing that the idea of constructing control variates by means of approximate solutions of an appropriate Poisson equation forms the basis of the so-called zero variance Markov Chain Monte Carlo methodology~\cite{papamarkou_al_2014}.
The estimator can be further modified by replacing the expectation in~\eqref{eq:simple_estimator_improvement_1},
which is intractable analytically,
by its value in the limit as $T \to \infty$;
that is, we define
\begin{equation}
    \label{eq:improved_estimator}
    v(T) =  \frac{\bigl\lvert \vect e^\t(\vect q_T - \vect q_0) \bigr\rvert^2}{2T} - \frac{\bigl\lvert \xi_T \bigr\rvert^2}{2T} + \lim_{T \to \infty} \expect \left[\frac{\bigl\lvert \xi_T \bigr\rvert^2}{2T}\right],
\end{equation}
Note that $v(T) = u(T)$ if $\psi_{\vect e} = 0$.
The expectation of $v(T)$ is different from that of $u(T)$,
but the two expectations coincide asymptotically as $T \to \infty$.
Furthermore, unlike the expectation in~\eqref{eq:simple_estimator_improvement_1},
the limit in the last term on the right-hand side of~\eqref{eq:improved_estimator} can be calculated explicitly,
and so the estimator $v(T)$ can be employed in practice.
\begin{lemma}
    \label{lemma:explicit_limit}
    Assume that $\psi_{\vect e} \in L^2(\mu)$ and $\grad_{\vect p} \psi_{\vect e} \in L^2(\mu)$.
    Then
    \begin{equation}
        \label{eq:explicit_limit}
        \lim_{T \to \infty} \expect \left[\frac{\bigl\lvert \xi_T \bigr\rvert^2}{2T}\right] = \gamma \beta^{-1} \int_{\torus^d \times \real^d} \abs{\grad_{\vect p} \psi_{\vect e}}^2 \, \d \mu =: d[\psi_{\vect e}].
    \end{equation}
\end{lemma}
\begin{proof}
    Let us introduce the notation
    \[
        \theta_T = \psi_{\vect e}(\vect q_0, \vect p_0) - \psi_{\vect e}(\vect q_T, \vect p_T),
        \qquad
        M_T = \sqrt{2 \gamma \beta^{-1}} \int_{0}^{T} \grad_{\vect p} \psi_{\vect e}(\vect q_t, \vect p_t) \cdot \d \vect w_t.
    \]
    From the definition~\eqref{eq:definition_control_variate},
    we have
    \[
        \frac{\lvert \xi_T \rvert^2}{2T} = \frac{\theta_T^2}{2T} + \frac{M_T^2}{2T} + \left(\frac{\theta_T}{\sqrt{T}}\right) \left(\frac{M_T}{\sqrt{T}}\right).
    \]
    Given that $\psi_{\vect e} \in L^2(\mu)$ and that we assume stationary initial conditions,
    so that $(\vect q_T, \vect p_T) \sim \mu$ as well,
    the expectation of the first term tends to 0 in the limit as $T \to 0$.
    The expectation of the second term can be calculated from It\^o's isometry:
    \[
        \expect \left[ \frac{M_T^2}{2T} \right]
        = \frac{\gamma \beta^{-1}}{T} \int_{0}^{T} \expect \Bigl[\left\lvert \grad_{\vect p} \psi_{\vect e}(\vect q_t, \vect p_t) \right\rvert^2\Bigr] \, \d t
        = \gamma \beta^{-1} \int_{\torus^d \times \real^d} \left\lvert \grad_{\vect p} \psi_{\vect e}(\vect q, \vect p) \right\rvert^2 \, \d \mu
        = d[\psi_{\vect e}].
    \]
    The expectation of the third term converges to zero by the Cauchy--Schwarz inequality,
    which concludes the proof of~\eqref{eq:explicit_limit}.
\end{proof}
Repeating verbatim the reasoning in the proof of~\cref{lemma:explicit_limit} with~$\phi_{\vect e}$ instead of $\psi_{\vect e}$ and $\vect e^\t (\vect q_T - \vect q_0)$ instead of $\xi_T$
(see~\eqref{eq:ito_for_phi}),
we obtain that
\[
    \lim_{T \to \infty} \expect [u(T)] = d[\phi_{\vect e}],
\]
implying that $d[\phi_{\vect e}] = D^{\gamma}_{\vect e}$,
since the limit on the left-hand side of this equation is by definition~$D^{\gamma}_{\vect e}$ in view of~\eqref{eq:einsteins_formula}.
This equality can also be shown from~\eqref{eq:effective_diffusion_poisson} by integrating by parts in the formula for $d[\phi_{\vect e}]$:
\begin{align}
    \notag
    d[\phi_{\vect e}]
    &= \gamma \beta^{-1} \int_{\torus^d \times \real^d} \abs{\grad_{\vect p} \phi_{\vect e}}^2 \, \d \mu
    = \gamma \beta^{-1} \int_{\torus^d \times \real^d} \bigl((\beta \grad V - \grad_{\vect p}) \cdot \grad_{\vect p}\phi_{\vect e} \bigr) \phi_{\vect e} \, \d \mu  \\
    \label{eq:equivalent_definition_effective_diffusion}
    &= -\int_{\torus^d \times \real^d} (\gamma \mathcal L_{\rm FD} \phi_{\vect e}) \phi_{\vect e} \, \d \mu
    = -\int_{\torus^d \times \real^d} (\mathcal L \phi_{\vect e}) \phi_{\vect e} \, \d \mu
    = D^{\gamma}_{\vect e},
\end{align}
where the skew-symmetry of $\mathcal L_{\rm ham}$ in $L^2(\mu)$ is employed in the second line.

By construction, it is clear that the improved estimator~\eqref{eq:improved_estimator} is asymptotically unbiased.
If~$\psi_{\vect e} = \phi_{\vect e}$, then this estimator is unbiased also for finite~$T$.
By a slight abuse of terminology,
we refer to the process $(\xi_t)_{t \geq 0}$ as the \emph{control variate} in the rest of this work.

\begin{remark}
    \label{remark:cost_control_variate}
    Notice that calculating the control variate $\xi_T$ in~\eqref{eq:definition_control_variate} requires
    to evaluate $\psi(q_t, p_t)$ at times 0 and $T$ and the gradient $\grad_{\vect p} \psi_{\vect e}(\vect q_t, \vect p_t)$ along the full trajectory $(\vect q_t, \vect p_t)_{0\leq t\leq T}$.
    Therefore, it is important for efficiency that $\grad_{\vect p} \psi_{\vect e}$ is not computationally expensive to evaluate.
\end{remark}

In the next subsections,
we obtain non-asymptotic results on the bias of the estimator $v(T)$ in~\cref{sub:bias},
and bounds on its variance in~\cref{sub:variance}.
Before this,
in order to build intuition and motivate our results,
we scrutinize two settings where
explicit expressions of the bias and variance of the estimator $u(T)$ can be obtained:
constant potential and quadratic potential (for systems in $\real^d$ rather than $\torus^d$).
In the rest of this section,
we employ the notation $\e^{\mathcal L t}$ to denote the Markov semigroup corresponding to the stochastic dynamics~\eqref{eq:langevin}:
\[
    \left(\e^{\mathcal L t} \varphi\right) (\vect q, \vect p) = \expect \bigl(\varphi(\vect q_t, \vect p_t) \big| (\vect q_0, \vect p_0) = (\vect q, \vect p)\bigr).
\]
\begin{example}
    [Constant potential]
    \label{example:constant}
    Consider the case where $V(q) = 0$ in dimension $d = 1$
    (henceforth we drop the $\vect e$ subscript and the bold notation for $\vect q$ and $\vect p$).
    In this case, the solution to the Poisson equation $- \mathcal L \phi = p$ is given by $\phi(q, p) = \gamma^{-1} p$,
    and applying Itô's formula to this function we obtain
    (this also follows directly from a time integration of~\eqref{eq:langevin_p})
    \[
        \gamma^{-1}(p_t - p_0) = - \int_{0}^{t} p_s \, \d s + \sqrt{2 \gamma^{-1} \beta^{-1}} w_t
        = q_0 - q_t + \sqrt{2 \gamma^{-1} \beta^{-1}} w_t.
    \]
    Using the explicit solution to the Ornstein--Uhlenbeck equation satisfied by $p$,
    we deduce that
    \begin{align*}
        q_t - q_0
        &= - \gamma^{-1} \left( p_0 \left(\e^{-\gamma t} - 1\right) + \sqrt{2 \gamma \beta^{-1}}\int_{0}^{t} \e^{-\gamma (t - s)} \, \d w_s \right)
        + \sqrt{2 \gamma^{-1} \beta^{-1}} w_t \\
        &=  - \gamma^{-1} p_0 \left(\e^{-\gamma t} - 1\right) + \sqrt{2 \gamma^{-1} \beta^{-1}}\int_{0}^{t} \left(1 - \e^{-\gamma (t - s)}\right) \, \d w_s.
    \end{align*}
    The assumptions on the initial condition imply that $p_0 \sim \mathcal N(0, \beta^{-1})$ and that $p_0$ is independent of $(w_t)_{t \geq 0}$,
    so the right-hand side of this equation is a mean-zero Gaussian random variable.
    Using It\^o's isometry, we calculate that $\expect \bigl[ u(T) \bigr]$ is given by
    \begin{align*}
        \frac{\expect \bigl[\abs{q_T - q_0}^2\bigr]}{2T}
        &= \frac{\lvert \e^{-\gamma T} - 1 \rvert^2 + 2 \gamma T - 4 (1 - \e^{-\gamma T}) +  1 - \e^{-2 \gamma T}}{2 \gamma^2 \beta T} \\
        &= \frac{1}{\gamma \beta} \left( 1 + \frac{1}{T \gamma} \left(\e^{-\gamma T} - 1\right) \right) =: \sigma_T^2.
    \end{align*}
    This equation implies that the effective diffusion coefficient in this example is $D^{\gamma} = \gamma^{-1} \beta^{-1}$,
    and that the relative bias is bounded from above by $(T \gamma)^{-1}$.
    Furthermore,
    since $\frac{\lvert q_T- q_0 \rvert^2}{2T\sigma_T^2}$ is distributed according to $\chi^2(1)$,
    the variance of $u(T)$ is equal to
    \[
        \var \bigl[u(T)\bigr] = 2  \bigl(\expect [u(T)]\bigr)^2 = 2 \sigma_T^4 \xrightarrow[T \to \infty]{} 2 \lvert D^{\gamma} \rvert^2 .
    \]
    Note that this variance does not converge to 0 as $T \to \infty$,
    a result further made precise for generic potentials in \cref{proposition:variance,proposition:asymptotic_variance}.
\end{example}

The case of a confining quadratic potential is degenerate,
in the sense that the associated effective diffusion coefficient is zero.
In this example,
we also obtain an explicit expression for the velocity autocorrelation function,
in order to motivate \cref{proposition:semigroup_meanzero_observable} below.
\begin{example}
    [Quadratic potential]
    \label{example:quadratic}
    We now consider the case of the one-dimensional quadratic confining potential $V(q) = \frac{k q^2}{2}$,
    and assume for simplicity $\gamma^2 - 4 k \neq 0$.
    In this case, the eigenfunctions of $\mathcal L$ are polynomials
    and, for every $n \geq 0$, the vector space $\poly(n)$ of polynomials of degree less than or equal to $n$
    contains an orthonormal basis of eigenfunctions of $\mathcal L$~\cite[Section~6.3]{pavliotis2011applied}.
    In particular, the constant function is an eigenfunction with eigenvalue 0,
    and the two other eigenfunctions in~$\poly(1)$, together with their associated eigenvalues,
    are given by
    \begin{equation*}
        g_{\pm}(q, p) =
        - \lambda_{\mp} q + p, \\
        \qquad
        \lambda_{\pm} = \frac{- \gamma \pm \sqrt{\gamma^2 - 4k}}{2}.
    \end{equation*}
    Here the radical symbol $\sqrt{\cdot}$ denotes the principal square root;
    for a complex number $z$, this is defined as $\sqrt{z} = \sqrt{r} \, \e^{\imag \theta/2}$ where $(r, \theta) \in [0, \infty) \times (-\pi, \pi]$ are the polar coordinates of~$z$.
    The coordinate functions $(q, p) \mapsto q$ and $(q, p) \mapsto p$ are the following linear combinations of $g_+$ and~$g_-$:
    \[
        q = \frac{g_+(q,p) - g_-(q,p)}{\lambda_+ - \lambda_-},
        \qquad
        p = \frac{\lambda_+ g_+(q,p) - \lambda_- g_-(q,p)}{\lambda_+ - \lambda_-}.
    \]
    Therefore,
    using the assumption that $(q_0, p_0) \sim \mu$,
    we have
    \begin{align*}
        \expect \bigl[ \lvert q_T - q_0 \rvert^2 \bigr]
        &= 2 \norm{q}^2 - 2 \ip{\e^{T \mathcal L}q}{q}
        = 2 \norm{q}^2 - 2 \frac{\ip{\e^{\lambda_+ T} g_+ - \e^{\lambda_- T} g_-}{q}}{\lambda_+ - \lambda_-} \\
        &= \frac{2}{k \beta} \left( 1 +  \frac{\lambda_- \e^{\lambda_+ T} - \lambda_+ \e^{\lambda_- T}}{\lambda_+ - \lambda_-} \right).
    \end{align*}
    This implies that
    \(
        T \expect \bigl[ u(T) \bigr] \to (k \beta)^{-1}
    \)
    in the limit as~$T \to \infty$,
    and so $D^{\gamma} = 0$ as expected.
    Similarly, it is not difficult to show $T^2\var \bigl[ u(T) \bigr] \to 2 (k \beta)^{-2}$ in the same limit;
    in this case, the variance is 0 asymptotically.
    Using that $\ip{g_+}{p} = \ip{g_-}{p} = \beta^{-1}$,
    we can also calculate the velocity autocorrelation function:
    \begin{equation}
        \label{eq:velocity_autocorrelation_quadratic}
        \ip{\e^{t \mathcal L}p}{p} =
        \frac{\lambda_+ \e^{\lambda_+ t} - \lambda_- \e^{\lambda_- t}}{\beta(\lambda_+ - \lambda_-)}.
    \end{equation}
    In the limit as $\gamma \to \infty$,
    it holds that $\lambda_+ \sim - k/\gamma$ and $\lambda_- \sim - \gamma$.
    In this limit,
    the factor multiplying the slowly decaying exponential $\e^{\lambda_+ t}$ in~\eqref{eq:velocity_autocorrelation_quadratic} scales as~$\mathcal O(\gamma^{-2})$,
    whereas the factor multiplying the rapidly decaying exponential $\e^{\lambda_- t}$ scales as $\mathcal O(1)$.
    We demonstrate in \cref{proposition:semigroup_meanzero_observable} that a similar property holds more generally.
\end{example}

\subsection{Bias of the estimators for the effective diffusion coefficient}%
\label{sub:bias}

In this subsection,
we begin by studying the bias of the standard estimator $u(T)$ in~\cref{ssub:bias_of_the_standard_estimator},
and then the bias of the improved estimator $v(T)$ in~\cref{ssub:bias_of_the_improved_estimator}.
Although we use, in \cref{sec:application_to_one_dimensional_langevin_type_dynamics,sec:applications_2d},
approximate solutions $\psi_{\vect e}$ of the Poisson equation~\eqref{eq:poisson_equation} that are not twice differentiable,
we focus in this section on the case where $\psi_{\vect e}$ is at least twice differentiable for simplicity of the analysis.

\subsubsection{Bias of the standard estimator}%
\label{ssub:bias_of_the_standard_estimator}

We first obtain in~\cref{lemma:easy_lemma} a simple bound on the bias based on standard results in the literature.
We then motivate, with the help of~\cref{example:quadratic},
that this result is not optimal in the overdamped regime and,
after obtaining a decay estimate for correlation functions of the form $t \mapsto \ip{\e^{t \mathcal L} f}{h}$ with $f$ and $h$ functions depending only on $p$,
we prove a finer bound on the bias in~\cref{corollary:better_bias}.

\begin{lemma}
    [Preliminary bound on the bias of the standard estimator]
    \label{lemma:easy_lemma}
    There exists a positive constant~$C$such that
    \begin{equation}
        \label{eq:bias_simple}
        \forall \gamma \in (0, \infty),
        \quad
        \forall T > 0,
        \qquad
        \abs{\expect \bigl[u(T)\bigr] - D_{\vect e}^{\gamma}}
        \leq
        \frac{C \max \{\gamma^{-2}, \gamma^2\}}{\beta T}.
    \end{equation}
\end{lemma}

\begin{proof}
Since the initial conditions are assumed statistically stationary,
it holds that
\begin{align}
    \notag
    \expect \bigl[u(T)\bigr]
    &= \frac{1}{2T} \expect \left[ \int_{0}^{T} \vect e^\t \vect p_{t} \, \d t \int_{0}^{T} \vect e^\t \vect p_s \, \d s \right]
    = \frac{1}{2T}  \int_{0}^{T} \!\! \int_{0}^{T} \expect \left[ \left(\vect e^\t \vect p_t\right) \left(\vect e^\t \vect p_s\right) \right] \, \d s \, \d t  \\
    \label{eq:first_equation_bias}
    &= \frac{1}{T}  \int_{0}^{T} \!\! \int_{0}^{t} \expect \left[ \left(\vect e^\t \vect p_t\right) \left(\vect e^\t \vect p_s\right) \right] \, \d s \, \d t,
\end{align}
since the contribution of the domain $0 \leq t \leq s \leq T$ is the same as that of $0 \leq s \leq t \leq T$.
The stationarity of the velocity process implies that, for ~$t \geq s$,
\begin{align*}
    \expect \left[ \left(\vect e^\t \vect p_t\right) \left(\vect e^\t \vect p_s\right) \right]
    &= \expect \left[ \left(\vect e^\t \vect p_{t-s}\right) \left(\vect e^\t \vect p_0\right) \right]
    = \ip{\e^{(t-s) \mathcal L} (\vect e^\t \vect p)}{\vect e^\t \vect p}.
\end{align*}
Substituting this expression in~\eqref{eq:first_equation_bias} and letting $\theta = t-s$ leads to
\begin{equation}
\label{eq:bias_without_control}
\begin{aligned}[b]
    \expect \bigl[u(T)\bigr]
    &= \int_{0}^{T} \ip{\e^{\theta \mathcal L} (\vect e^\t \vect p)}{\vect e^\t \vect p} \left(1 - \frac{\theta}{T}\right) \d \theta  \\
    &= \int_{0}^{\infty} \ip{\e^{\theta \mathcal L}(\vect e^\t \vect p)}{\vect e^\t \vect p}  \d \theta
    - \int_{0}^{\infty} \ip{\e^{\theta \mathcal L} (\vect e^\t \vect p)}{\vect e^\t \vect p} \min\left\{1, \frac{\theta}{T}\right\} \, \d \theta.
\end{aligned}
\end{equation}
As we shall demonstrate, the second term tends to 0 in the limit as $T \to \infty$.
Therefore, since the estimator $u(T)$ is asymptotically unbiased,
the first term must coincide with the effective diffusion coefficient~$D^{\gamma}_{\vect e}$
-- this is in fact the well known Green--Kubo formula for the effective diffusion coefficient, see e.g.~\cite{pavliotis2011applied,MR3509213}.
The Green--Kubo formula can also be derived from~\eqref{eq:effective_diffusion_poisson} by using the representation formula $\phi_{\vect e} = \int_{0}^{\infty} \e^{\theta \mathcal L} \left(\vect e^\t \vect p\right) \, \d \theta$,
which is well defined in view of the exponential decay of $\e^{\theta \mathcal L}$ on $L^2_0(\mu)$,
see~\eqref{eq:decay_semigroup_general} below.
The second term in~\eqref{eq:bias_without_control} is the bias.
In order to bound this term,
we use a general bound for the Markov semigroup associated with Langevin dynamics
stating that
\begin{equation}
    \label{eq:decay_semigroup_general}
    \forall \gamma > 0, \qquad \forall \theta \geq 0, \qquad
    \norm*{ \e^{\theta \mathcal L} }[\mathcal B \left(L^2_0\left(\mu\right) \right)] \leq L \exp \bigl(- \ell \theta \min\{\gamma, \gamma^{-1}\} \bigr)
\end{equation}
for appropriate constants $L > 0$ and $\ell > 0$.
Here $\mathcal B\left(L^2_0(\mu)\right)$ is the Banach space of bounded linear operators on $L^2_0(\mu)$,
and $\norm{\cdot}_{\mathcal B\left(L^2_0(\mu)\right)}$ is the usual associated norm.
This result is proved in~\cite{MR2394704} for $\gamma \in (0, 1)$ using the $H^1$ hypocoercivity approach~\cite{MR2562709},
and later in~\cite{MR3106879} for general $\gamma \in (0, \infty)$, in the Fokker--Planck setting,
using the direct $L^2(\mu)$ hypocoercivity approach pioneered in~\cite{Herau06,MR2576899,MR3324910}.
The latter approach is revisited in the backward Kolmogorov setting in~\cite{MR3522857,roussel2018spectral}.
An application of the bound~\eqref{eq:decay_semigroup_general} gives
\begin{equation}
    \label{eq:initial_bound_velocity_autocorrelation}
    \left\lvert \ip{\e^{\theta \mathcal L}(\vect e^\t p)}{\vect e^\t p} \right\rvert
    \leq \norm{\e^{\theta \mathcal L}(\vect e^\t p)} \norm{\vect e^\t p}
    \leq L \beta^{-1} \exp\bigl(- \ell \theta \min\{\gamma, \gamma^{-1}\}\bigr).
\end{equation}
Noting that
\begin{equation}
    \label{eq:petit_calcul}
    \forall \lambda > 0, \qquad
    \int_{0}^{\infty} \e^{-\lambda \theta} \min\left\{ 1, \frac{\theta}{T} \right\} \, \d \theta
    = \frac{1 - \e^{- \lambda T}}{\lambda^2 T} \leq \frac{1}{\lambda^2 T},
\end{equation}
we obtain~\eqref{eq:bias_simple}.
\end{proof}
Since the effective diffusion coefficient scales as $\gamma^{-1}$ in both the underdamped ($\gamma \to 0$) and overdamped limits ($\gamma \to \infty$)~\cite{MR2394704,MR2427108},
this estimate~\eqref{eq:bias_simple} suggests that the relative bias of the estimator scales as $\max\{\gamma^{-1}, \gamma^3\} T^{-1}$ and that,
consequently, the integration time $T$ should scale proportionally to $\max\{\gamma^{-1}, \gamma^3\}$ in order to achieve a given relative accuracy.
It turns out that the estimate~\eqref{eq:initial_bound_velocity_autocorrelation} is not optimal in the overdamped regime,
which is clear in the case of quadratic potential; see~\eqref{eq:velocity_autocorrelation_quadratic} in \cref{example:quadratic}.
We derive a sharper estimate from the following~\cref{proposition:semigroup_meanzero_observable}.
In order to state this result,
we introduce the operators~$\Pi_{\vect p}\colon \lp{2}{\mu} \rightarrow \lp{2}{\mu}$ and  $\Pi_{\vect p}^\perp = \id - \Pi_{\vect p}$,
with
\[
    \Pi_{\vect p} u (\vect q) = \int u(\vect q, \vect p) \, \kappa(\d \vect p).
\]
The operators $\Pi_{\vect p}$ and $\Pi_{\vect p}^\perp$ are respectively the $L^2(\mu)$ projection operators onto
the subspace of functions depending only on $\vect q$,
and the subspace of functions with average $0$ in $\vect p$ (with respect to the marginal distribution $\kappa$, defined in~\eqref{eq:definition_prob_measures}, and for almost every $\vect q \in \torus^d$).
We also introduce the space $H^{1,\vect q}(\mu)$ of functions in $L^2(\mu)$ with their $\vect q$-gradient also in $L^2(\mu)$,
and the associated norm $\norm{\dummy}[1,\vect q] = \norm{\dummy} + \norm{\grad_{\vect q} \dummy}$.
\begin{proposition}
    \label{proposition:semigroup_meanzero_observable}
    Assume that $f \in H^{1,\vect q}(\mu)$ and $h \in H^{1,\vect q}(\mu)$ are smooth functions in~$\Pi_{\vect p}^\perp \lp{2}{\mu}$.
    Then there exist positive constants $A$ and $a$, independent of $f$ and $h$, such that
    \begin{equation}
        \label{eq:optimal_decay_correlation}
        \forall \gamma \geq 1, \quad
        \forall t \geq 0, \qquad
        \abs{\ip{\e^{t \mathcal L}f}{h}}
        \leq A \norm{f}[1,\vect q]  \norm{h}[1,\vect q] \left( \gamma^{-2} \e^{- a \gamma^{-1} t} + \e^{- a \gamma t} \right).
    \end{equation}
\end{proposition}

This result, proved in \cref{sec:auxiliary_technical_results},
enables to show the following bound on the bias of~$u(T)$,
which is better than~\cref{lemma:easy_lemma} in the large $\gamma$ regime.
Roughly speaking, \cref{proposition:semigroup_meanzero_observable} states that,
when $\gamma \gg 1$ and $f$ and $h$ are mean-zero in $\vect p$,
correlations of the form $\expect \bigl[ f(\vect p_t) h(\vect p_0) \bigr]$ are~$\mathcal O(\gamma^{-2})$ small after a small time of order $\mathcal O(\gamma^{-1} \log \gamma)$,
despite the fact that their asymptotic decay as~$\e^{- a \gamma^{-1} t}$ is slow.
\begin{corollary}
    [Bias of the standard estimator]
    \label{corollary:better_bias}
    There exists a positive constant~$\widehat C$ such that
    \begin{equation}
        \label{eq:bias}
        \forall \gamma \in (0, \infty),
        \quad
        \forall T > 0,
        \qquad
        \abs{\expect \bigl[u(T)\bigr] - D_{\vect e}^{\gamma}}
        \leq
        \frac{\widehat C \max \{\gamma^{-2}, 1\}}{\beta T}.
    \end{equation}
\end{corollary}
\begin{proof}
    Applying \cref{proposition:semigroup_meanzero_observable} with $f(\vect q, \vect p) = h(\vect q, \vect p) = \vect e^\t \vect p$,
    and recalling that the bias coincides with the second term on the right-hand side of~\eqref{eq:bias_without_control},
    we obtain
    \begin{align}
        \label{eq:refined_bound_for_u}
        \abs{\expect \bigl[u(T)\bigr] - D^{\gamma}_{\vect e}}
        &\leq \frac{A}{\beta} \int_{0}^{\infty} \left( \gamma^{-2} \e^{- a \gamma^{-1} \theta} + \e^{-a  \gamma \theta} \right)  \min \left\{1,  \frac{\theta}{T} \right\} \, \d \theta \\
        &\leq \frac{A}{\beta T} \int_{0}^{\infty} \left( \gamma^{-2} \e^{- a \gamma^{-1} \theta} + \e^{-a  \gamma \theta} \right)  \theta \, \d \theta
        \leq \frac{A}{\beta a^2 T} \left(1 + \frac{1}{\gamma^2} \right),
    \end{align}
    which directly yields the result.
\end{proof}
The estimate~\eqref{eq:bias} shows that the relative bias in fact scales as~$\max\{\gamma^{-1}, \gamma\} T^{-1}$,
and so it is sufficient to take~$T \propto \gamma$ in order to control the bias in the overdamped limit.

\begin{remark}
The case where~$V(\vect q) = 0$ is particular,
in that the correlation $\ip{\e^{\theta \mathcal L} (\vect e^\t \vect p)}{\vect e^\t \vect p}$ decays as $\e^{-\gamma t}$ with a prefactor independent of $\gamma$ in this setting.
Consequently, the bias of~$u(T)$ scales as $(\gamma T)^{-1}$ in both the underdamped and the overdamped regimes,
as observed in~\cref{example:constant}.
\end{remark}

\subsubsection{Bias of the improved estimator}%
\label{ssub:bias_of_the_improved_estimator}
We now obtain a bound on the bias of the improved estimator $v(T)$.
The following result can be viewed as a generalization of~\cref{lemma:easy_lemma},
which is recovered in the particular case when $\psi_{\vect e} = 0$.
\begin{proposition}
    [Bias of the estimator]
    \label{lemma:bias_improved}
    Assume that $\mathcal L \psi_{\vect e} \in \lp{2}{\mu}$.
    With the same notation as in~\eqref{eq:decay_semigroup_general},
    it holds that
    \begin{align}
        \label{eq:basic_bound_bias}
        \forall \gamma \in (0, \infty), \qquad
        \abs{\expect \bigl[ v(T) \bigr] - D^{\gamma}_{\vect e}}
                &\leq  \frac{L \max\{\gamma^2, \gamma^{-2}\}}{T \ell^2 }  \,  \norm{\vect e^\t \vect p + \mathcal L \psi_{\vect e}}  \left(\beta^{-1/2} + \norm{\mathcal L \psi_{\vect e}} \right).
    \end{align}
\end{proposition}
Note that the right-hand side of~\eqref{eq:basic_bound_bias} is small when $\mathcal L \psi_{\vect e} \approx \mathcal L \phi_{\vect e} = - \vect e^\t p$.
\begin{proof}
    Using Itô's formula for $\psi_{\vect e}$,
    we have
    \[
        \psi_{\vect e}(\vect q_T, \vect p_T) - \psi_{\vect e}(\vect q_0, \vect p_0)
        = \int_{0}^{T} (\mathcal L \psi_{\vect e}) (\vect q_t, \vect p_t) \, \d t
        + \sqrt{2 \gamma \beta^{-1}} \int_{0}^{T} \grad_{\vect p} \psi_{\vect e} (\vect q_t, \vect p_t) \cdot \d \vect w_t,
    \]
    and employing the same reasoning as in~\eqref{eq:bias_without_control}, we obtain
    \begin{align*}
        \expect \bigl[v(T)\bigr]
    &= d[\psi_{\vect e}] +  \frac{1}{2T} \, \expect \biggl[ \bigl\lvert \vect e^\t (\vect q_T - \vect q_0) \bigr\rvert^2 - \biggl| \int_0^T {\mathcal L \psi_{\vect e}}(\vect q_t, \vect p_t) \, \d t \biggr|^2 \biggr] \\
    &= d[\psi_{\vect e}] +  \int_{0}^{T} \Bigl( \ip{\e^{\theta \mathcal L}\bigl(\vect e^\t \vect p\bigr)}{\vect e^\t \vect p} - \ip{\e^{\theta \mathcal L} \mathcal L \psi_{\vect e}}{\mathcal L \psi_{\vect e}} \Bigr) \left( 1 - \frac{\theta}{T} \right) \d \theta \\
    &= D^{\gamma}_{\vect e} - \int_{0}^{\infty} \min\left\{1, \frac{\theta}{T}\right\} \Bigl( \ip{\e^{\theta \mathcal L}\bigl(\vect e^\t \vect p\bigr)}{\vect e^\t  \vect p} - \ip{\e^{\theta \mathcal L} \mathcal L \psi_{\vect e}}{\mathcal L \psi_{\vect e}} \Bigr) \, \d \theta.
    \end{align*}
    we denote the $L^2(\mu)$ of the generator $\mathcal L$ by~
    \(
        \mathcal L^* = - \mathcal L_{\rm Ham} + \gamma \mathcal L_{\rm FD}
    \).
    We have
    \begin{align*}
         \left\lvert \ip{\e^{t \mathcal L}(\vect e^\t \vect p)}{\vect e^\t \vect p} - \ip{\e^{t \mathcal L} \mathcal L \psi_{\vect e}}{\mathcal L \psi_{\vect e}} \right\rvert
        &= \abs{\ip{\e^{t \mathcal L} (\vect e^\t \vect p)}{\vect e^\t \vect p + \mathcal L \psi_{\vect e}}- \ip{\e^{t \mathcal L} \left(\vect e^\t \vect p + \mathcal L \psi_{\vect e}\right)}{\mathcal  L\psi_{\vect e}}} \\
    &\qquad \leq \norm{\e^{t \mathcal L}}[\mathcal B\left(L^2_0(\mu) \right)] \norm*{\vect e^\t \vect p + \mathcal L \psi_{\vect e}}
     \left(\norm*{\vect e^\t \vect p} + \norm{\mathcal L \psi_{\vect e}} \right) \\
    &\qquad \leq L \e^{- \ell \min\{\gamma, \gamma^{-1}\} t} \norm*{\vect e^\t \vect p + \mathcal L\psi_{\vect e}}  \left(\beta^{-1/2} + \norm{\mathcal L \psi_{\vect e}} \right),
    \end{align*}
    where $L$ and $\ell$ are the same constants as in~\eqref{eq:decay_semigroup_general}.
    We finally obtain~\eqref{eq:basic_bound_bias} in view of~\eqref{eq:petit_calcul}.
\end{proof}

\Cref{lemma:bias_improved} suffers from the same shortcoming as~\cref{lemma:easy_lemma}:
it is not optimal in the large $\gamma$ regime.
Employing~\cref{proposition:semigroup_meanzero_observable} in a similar manner as in the proof of~\cref{corollary:better_bias},
we prove in~\cref{sec:proof_technical_result} that,
if $\mathcal L \psi_{\vect e} \in H^{1,\vect q}(\mu)$,
then there is $C$ independent of $\psi_{\vect e}$ such that
\begin{align}
    \notag
    \forall \gamma \geq 1, \quad
    \forall T > 0, \qquad
    \abs{\expect \bigl[ v(T) \bigr] - D_{\vect e}^{\gamma}}
        &\leq C T^{-1}
        \norm*{\vect e^\t \vect p +  \mathcal L \psi_{\vect e}}[1,\vect q] \, \bigl(\beta^{-1/2} + \norm{\mathcal L \psi_{\vect e}}[1,\vect q] \bigr) \\
        \label{eq:refined_bound}
        &\quad + C T^{-1} \gamma^2 \norm{\Pi_{\vect p} \mathcal L \psi_{\vect e}} \norm{\mathcal L \psi_{\vect e}},
\end{align}
This bound is not as satisfying as~\cref{corollary:better_bias},
because a $\gamma^2$ factor remains in the second term on the right-hand side,
although the prefactor $\norm{\Pi_{\vect p} \mathcal  L \psi_{\vect e}}$ is expected to be small as~$\norm{\Pi_{\vect p} \mathcal L \psi_{\vect e}} = \norm{\Pi_{\vect p} (\vect e^\t \vect p + \mathcal L \psi_{\vect e})} \leq \norm{\vect e^\t \vect p + \mathcal L \psi_{\vect e}}$.
The bound~\eqref{eq:refined_bound} is therefore an improvement over~\eqref{eq:basic_bound_bias} for large $\gamma$.
However,
unless $\Pi_{\vect p} \mathcal L \psi_{\vect e} = \Pi_{\vect p} \mathcal L_{\rm ham} \psi_{\vect e} = \mathcal O(\gamma^{-2})$,
the dependence on $\gamma$ of the bias in~\eqref{eq:refined_bound} is worse in the limit $\gamma \to \infty$
than that of the simple estimator $u(T)$, see~\eqref{eq:bias}.
It is then not clear that employing a control variate is useful in this limit.
Since our focus in this work is on the underdamped limit $\gamma \to 0$,
and since the overdamped limit $\gamma \to \infty$ for one or two-dimensional systems is more easily studied numerically through deterministic methods anyway,
we do not further investigate this issue.

\subsection{Variance of the estimators}%
\label{sub:variance}

In this section, we obtain bounds on the variance of the estimator $v(T)$,
first for finite $T$ and then asymptotically in the limit as~$T \to \infty$.
Since $v(T)$ and $u(T)$ coincide when $\psi_{\vect e} = 0$,
bounds on the variance of $u(T)$ can be recovered by letting $\psi_{\vect e} = 0$ in the estimates below.

Since it is difficult to obtain bounds that scale well both as $\gamma \to 0$ and as $\gamma \to \infty$,
we aim here at obtaining bounds with a good scaling only in the underdamped regime~$\gamma \to 0$,
as this is the regime where our approach is of practical interest.

\begin{proposition}
    \label{proposition:variance}
    There exists $C > 0$ independent of $\gamma$, $T$ and $\psi_{\vect e}$ such that
    \begin{equation}
        \label{eq:statement_variance}
        \begin{aligned}[b]
            \var \bigl[v(T)\bigr]
            \leq
            C &\left( T^{-1} \norm{\phi_{\vect e} - \psi_{\vect e}}[L^4(\mu)]^2  + \gamma \norm{\grad_{\vect p} \phi_{\vect e} - \grad_{\vect p} \psi_{\vect e}}[L^4(\mu)]^2 \right) \\
              &\quad \times \left( T^{-1} \norm{\phi_{\vect e} + \psi_{\vect e}}[L^4(\mu)]^2  + \gamma \norm{\grad_{\vect p} \phi_{\vect e} + \grad_{\vect p} \psi_{\vect e}}[L^4(\mu)]^2 \right),
        \end{aligned}
    \end{equation}
    provided that all the terms on the right-hand side are finite.
\end{proposition}
\begin{remark}
    As already observed in \cref{example:constant},
    the variance does not converge to zero in the limit as $T \to \infty$.
    This asymptotic behavior is in contrast with that of estimators based on linear response,
    but not unlike the behavior of estimators based on the Green-Kubo formula,
    where in fact the variance grows with the integration time~\cite{LMS16}.
    We study more precisely the behavior of the variance in the limit as $T \to \infty$ in \cref{proposition:asymptotic_variance} below.
\end{remark}
\begin{remark}
    In order to assess the quality of the upper bound~\eqref{eq:statement_variance} in the underdamped limit~$\gamma \to 0$,
    let us consider the particular setting where $\psi_{\vect e} = 0$ in one dimension.
    In~\cite[Remark 6.10]{MR2394704},
    it is proves that $\norm{\phi}[\lp{r}{\mu}] = \bigo{\gamma^{-1}}$ as $\gamma \to 0$ for every $r \in [1, \infty)$,
    and it is conjectured that also $\norm{\partial_q \phi}[\lp{r}{\mu}] = \mathcal O(\gamma^{-1})$ in the same limit.
    Assuming that this is true, the estimate~\eqref{eq:statement_variance} gives
    \begin{equation}
        \label{eq:variance_scaling}
        \forall \gamma \in (0, 1), \qquad
        \var \bigl[v(T)\bigr]
        \leq \widetilde C \left( T^{-1} \gamma^{-2} + \gamma^{-1} \right)^2,
    \end{equation}
    for some constant $\widetilde C > 0$.
    For an integration time $T$ scaling as $\gamma^{-1}$,
    which is required in order to control the bias,
    we find from this formula that the variance scales as $\gamma^{-2}$,
    and so the relative standard deviation scales as $\mathcal O(1)$.
    In practice, this means that we can keep the number of Monte Carlo replicas constant as $\gamma \to 0$
    without degrading the relative width of our confidence interval.
\end{remark}
\begin{remark}
    We choose $a = b = 2$ in the following proof,
    and so the norm obtained on the right-hand side of~\eqref{eq:statement_variance} is that of~$L^4(\mu)$.
    Naturally, other choices could have been considered.
\end{remark}
\begin{proof}
    The proof is based on the crude inequality
    \begin{align}
        \notag
        \var \bigl[v(T)\bigr]
        &= \min_{\mathscr V \in \real} \expect \Bigl[\bigl(v(T) - \mathscr V\bigr)^2\Bigr]
        \leq \expect \left[ \bigl( v(T) - d[\psi_{\vect e}] \bigr)^2 \right] \\
        &= \frac{1}{4 T^2} \expect \left[ \left\lvert \vect e^\t (\vect q_T - \vect q_0) - \xi_T \right\rvert^2 \, \left\lvert \vect e^\t (\vect q_T - \vect q_0) + \xi_T \right\rvert^2 \right] \\
        \label{eq:bound_variance_two_factors}
        &\leq \frac{1}{4 T^2}
        {\left(\expect \left[  \left\lvert \vect e^\t(\vect q_T - \vect q_0) - \xi_T \right\rvert^{2a} \right]\right)^{\frac{1}{a}}} \,
        {\left(\expect \left[ \left\lvert \vect e^\t (\vect q_T - \vect q_0) + \xi_T \right\rvert^{2b} \right]\right)^{\frac{1}{b}}},
    \end{align}
    for any $(a, b) \in (1, \infty)^2$ such that $\frac{1}{a} + \frac{1}{b} = 1$.
    We will use the notation
    \[
        I_{\phi} = \sqrt{2 \gamma \beta^{-1}} \int_{0}^{T} \, \grad_{\vect p} \phi_{\vect e} (\vect q_t, \vect p_t) \cdot \d \vect w_t,
        \qquad
        I_{\psi} = \sqrt{2 \gamma \beta^{-1}} \int_{0}^{T} \, \grad_{\vect p} \psi_{\vect e} (\vect q_t, \vect p_t) \cdot \d \vect w_t.
    \]
    Using the short-hand notations $\phi_{s} = \phi_{\vect e}(\vect q_s, \vect p_s)$ and $\psi_{s} = \psi_{\vect e}(\vect q_s, \vect p_s)$,
    we have by~\eqref{eq:ito_for_phi} and~\eqref{eq:definition_control_variate}
    \begin{align*}
        \expect \left[ |\vect e^\t(\vect q_T - \vect q_0) - \xi_T|^{2 a} \right]
        &= \expect \left[ \abs{\phi_{0} - \phi_{T} - \psi_{0} + \psi_{T}  + I_{\phi} - I_{\psi}}^{2 a} \right] \\
        &\leq 3^{2 a-1} \left( \expect \left[ \abs{\phi_T - \psi_T}^{2 a} \right] + \expect \left[ \abs{\phi_0 - \psi_0}^{2 a} \right] + \expect \left[ \abs{I_{\psi} - I_{\phi}}^{2 a} \right] \right),
    \end{align*}
    where we employed the inequality $|x_1 + \dotsb + x_N|^{y} \leq N^{y-1} \left( |x_1|^y + \dotsb + |x_N|^y \right)$ for any $y \geq 1$,
    which follows from the convexity of $x \mapsto \abs{x}^y$.
    The first two terms are equal to $\norm{\phi_{\vect e} - \psi_{\vect e}}_{L^{2 a}(\mu)}^{2 a}$
    given the assumption of stationary initial condition.
    Using a moment inequality for It\^o integrals~\cite[Theorem 7.1]{MR2380366},
    we bound the last term as
    \begin{align*}
        \expect \left[ \abs{I_{\psi} - I_{\phi}}^{2 a} \right]
        &\leq \bigl( a(2 a - 1)\bigr)^a \, \left(2 \gamma \beta^{-1}\right)^a \,  T^{a-1} \,
        \expect \left[ \int_{0}^{T} \abs{\grad_{\vect p} \phi_{\vect e}(\vect q_s, \vect p_s) - \grad_{\vect p} \psi_{\vect e}(\vect q_s, \vect p_s)}^{2a} \, \d s\right] \\
        &= \bigl(a(2a - 1)\bigr)^a \, \left(2 \gamma \beta^{-1}\right)^a \,  T^{a} \, \int_{\torus^d \times \real^d} \abs{\grad_{\vect p} \phi_{\vect e} - \grad_{\vect p} \psi_{\vect e}}^{2a} \d \mu \\
        &= \bigl(a(2a - 1)\bigr)^a \, \left(2 \gamma \beta^{-1}\right)^a \,  T^{a} \, \norm{\grad_{\vect p} \phi_{\vect e} - \grad_{\vect p} \psi_{\vect e}}_{L^{2a}(\mu)}^{2a}.
    \end{align*}
    Likewise, the second factor in~\eqref{eq:bound_variance_two_factors} can be bounded as
    \begin{align*}
        &\expect \left[ |\vect e^\t(\vect q_T - \vect q_0) + \xi_T|^{2b} \right] \\
        &\qquad
        \leq 3^{2b-1} \Bigl( 2 \norm{\phi_{\vect e} + \psi_{\vect e}}[L^{2b}(\mu)]^{2b}
        + \bigl(b(2b - 1)\bigr)^b \, \left(2 \gamma \beta^{-1}\right)^b \, T^{b} \norm{\grad_{\vect p} \phi_{\vect e} + \grad_{\vect p} \psi_{\vect e}}[L^{2b}(\mu)]^{2b} \Bigr).
    \end{align*}
    The statement is then obtained by choosing $a = b = 2$.
\end{proof}

To conclude this section,
we quantify more precisely the asymptotic variance of $v(T)$ in the limit as $T \to \infty$.
\begin{proposition}
    [Asymptotic value of the variance]
    \label{proposition:asymptotic_variance}
    Assume that there exists $\varepsilon > 0$ such that
    \begin{equation}
        \label{eq:regularity_solution}
        \phi_{\vect e} \in L^{4+\varepsilon}(\mu),
        \qquad
        \psi_{\vect e} \in L^{4+\varepsilon}(\mu),
        \qquad
        \grad_{\vect p} \phi_{\vect e} \in L^{4+\varepsilon}(\mu),
        \qquad
        \grad_{\vect p} \psi_{\vect e} \in L^{4+\varepsilon}(\mu).
    \end{equation}
    Then it holds that
    \begin{equation}
        \label{eq:asymptotic_variance}
        \lim_{T \to \infty} \var\bigl[v(T)\bigr] =
        2 \Bigl( d[\phi_{\vect e}]^2 +  d[\psi_{\vect e}]^2 \Bigr) - 4 \gamma^2 \beta^{-2} \left( \int \grad_{\vect p} \phi_{\vect e} \cdot \grad_{\vect p} \psi_{\vect e} \, \d \mu \right)^2.
    \end{equation}
    In particular,
    \begin{align}
        \label{eq:bounds_asymptotic_variance}
        2\bigl( d[\phi_{\vect e}] - d[\psi_{\vect e}] \bigr)^2
        \leq \lim_{T \to \infty} \var \bigl[v(T)\bigr]
        &\leq 4 \gamma \beta^{-1} \norm{\grad_{\vect p} \phi_{\vect e} - \grad_{\vect p} \psi_{\vect e}}^2 (d[\phi_{\vect e}] + d[\psi_{\vect e}]).
    \end{align}
\end{proposition}
\begin{remark}
    \label{remark:asym_variance_u}
    The result~\eqref{eq:asymptotic_variance} implies that $\var\bigl[v(T)\bigr] \to 2 d[\phi_{\vect e}]^2$ in the limit as $T \to \infty$ when $\psi_{\vect e} = 0$,
    which is consistent with our explicit computations in \cref{example:constant} for the case of a constant potential.
\end{remark}
\begin{proof}
    Using It\^o's isometry and the martingale central limit theorem~(see, for instance, \cite{MR668684} or \cite[Theorem 3.3]{pavliotis2008multiscale}),
    we obtain that
    {\small%
        \[
            X^T :=
            \frac{1}{\sqrt{2T}}
            \begin{pmatrix}
                \vect e^\t (q_T - q_0) \\
                \xi_T
            \end{pmatrix}
            \xrightarrow[T \to \infty]{\rm Law}
            X^{\infty} \sim
            \mathcal N \left(0,
                \gamma \beta^{-1}
                \begin{pmatrix}
                    \int \grad_{\vect p} \phi_{\vect e} \cdot \grad_{\vect p} \phi_{\vect e} \, \d \mu & \int \grad_{\vect p} \phi_{\vect e} \cdot \grad_{\vect p} \psi_{\vect e} \, \d \mu \\
                    \int \grad_{\vect p} \phi_{\vect e} \cdot \grad_{\vect p} \psi_{\vect e} \, \d \mu & \int \grad_{\vect p} \psi_{\vect e} \cdot \grad_{\vect p} \psi_{\vect e} \, \d \mu
                \end{pmatrix}
            \right),
        \]
    }
    where the domain of integration of all the integrals in the covariance matrix is $\torus^d \times \real^d$.
    For a bivariate Gaussian vector $X^{\infty} \sim \mathcal N(0, \Sigma)$ with density $g_{\Sigma}$,
    it holds by the general formula for the higher-order moments of multivariate Gaussians (Isserlis' theorem)
    that
    \begin{align}
        \notag
        \var \bigl[ \lvert X^{\infty}_1 \rvert^2 - \lvert X^{\infty}_2 \rvert^2\bigr]
        &= \int_{\real^2} \left( x_1^2 - x_2^2 - \Sigma_{11} + \Sigma_{22} \right)^2 \, g_{\Sigma}(x_1, x_2) \, \d x_1 \, \d x_2 \\
        \label{eq:intermediate_identity}
        &= 2 \Sigma_{11}^2 + 2 \Sigma_{22}^2 - 4 \Sigma_{12}^2.
    \end{align}
    We prove in~\cref{sec:proof_of_equation_variance} that,
    assuming~\eqref{eq:regularity_solution},
    \begin{equation}
        \label{eq:limit_variance}
        \var \Bigl[ \lvert X^T_1 \rvert^2 - \lvert X^T_2 \rvert^2 \Bigr]
        \xrightarrow[T \to \infty]{}
        \var \Bigl[ \lvert X^{\infty}_1 \rvert^2 - \lvert X^{\infty}_2 \rvert^2 \Bigr].
    \end{equation}
    (This equation does not follow directly from the convergence in distribution of $X^T$ to $X^{\infty}$,
    because polynomials are not uniformly bounded.)
    Combining~\eqref{eq:limit_variance} with the identity~\eqref{eq:intermediate_identity} directly implies~\eqref{eq:asymptotic_variance}.
    The lower bound in~\eqref{eq:bounds_asymptotic_variance} then follows from an application of the Cauchy--Schwarz inequality.
    In order to obtain the upper bound,
    we use that in any inner product space it holds
    \begin{align*}
        2 \norm{a}^4 + 2 \norm{b}^4 - 4\ip{a}{b}^2
        &= 2 \norm{a}^4 + 2 \norm{b}^4 - \bigl(\ip{a}{a} + \ip{b}{b} - \ip{a-b}{a-b} \bigr)^2 \\
        &= \bigl(\norm{a}^2 - \norm{b}^2 \bigr)^2 + 2 \norm{a-b}^2 \bigl(\norm{a}^2 + \norm{b}^2 \bigr) - \norm{a-b}^4 \\
        &= \bigl(\norm{a}^2 - \norm{b}^2 \bigr)^2 + \norm{a-b}^2 \norm{a + b}^2 \\
        &= \ip{a - b}{a + b}^2 + \norm{a-b}^2 \norm{a + b}^2 \\
        &\leq 2 \norm{a - b}^2 \norm{a + b}^2
        \leq 4 \norm{a - b}^2 \bigl(\norm{a}^2 + \norm{b}^2\bigr).
    \end{align*}
    The desired upper bound is obtained by
    using this inequality with $a = \grad_{\vect p} \phi_{\vect e}$ and $b = \grad_{\vect p} \psi_{\vect e}$ in the Hilbert space $L^2(\mu)$.
\end{proof}

\section{Application to one-dimensional Langevin-type dynamics}%
\label{sec:application_to_one_dimensional_langevin_type_dynamics}

We consider in this section two different approaches for constructing an approximate solution to the Poisson equation~\eqref{eq:poisson_equation}:
through a Fourier/Hermite Galerkin method in \cref{sub:galerkin_approach},
and through formal asymptotic expansions for the underdamped regime in \cref{sub:underdamped_approach}.
We then present numerical results in \cref{sub:numerical_results}
and discuss an extension of our approach to higher order Langevin dynamics, obtained as the Markovian approximation of the generalized Langevin equation, in \cref{sub:generalization_to_generalized_langevin_dynamics}.
Throughout this section, we consider the one-dimension potential
\[
    V(q) = - \frac{\cos(q)}{2}.
\]
Since we are concerned only with one-dimensional dynamics in most of this section,
we employ the scalar notation $q, p, \phi, \psi, D^{\gamma}$ in place of $\vect q, \vect p, \phi_{\vect e}, \psi_{\vect e}, D^{\gamma}_{\vect e}$.

\subsection{Fourier/Hermite spectral method}%
\label{sub:galerkin_approach}
We employ the non-conformal Galerkin method developed and analyzed in~\cite{roussel2018spectral}.
Specifically, we calculate an approximate solution to~\eqref{eq:poisson_equation} through the following saddle point formulation:
find~$\Psi_N \in V_N$ such that
\begin{align}
  \label{eq:saddle_point_formutation}
  \left\{
    \begin{aligned}
       & - \mathscr P_N \, \mathcal L \, \mathscr P_N \Psi_N + \alpha_N u_N = \mathscr P_N p, \\
       & \ip{\Psi_N}{u_N} = 0,
    \end{aligned}
  \right.
\end{align}
where $V_N$ is a finite-dimensional approximation space,
$\mathscr P_N$ is the $\lp{2}{\mu}$ projection operator onto~$V_N$,
$u_N = \mathscr P_N 1 / \norm{\mathscr P_N 1} \in V_N$
and $\alpha_N$ is a Lagrange multiplier.
As above,
$\ip{\cdot}{\cdot}$ and~$\norm{\cdot}$ denote respectively the standard inner product and norm of $\lp{2}{\mu}$.
The formulation~\eqref{eq:saddle_point_formutation} ensures
that the system is well-conditioned at the finite dimensional level.
The solution $\Psi_N$ to~\eqref{eq:saddle_point_formutation} equivalently solves
\[
    - \mathcal P_N \, \mathcal L \, \mathcal P_N \Psi_N = \mathcal P_N p, \qquad \Psi_N \in \{\phi - \ip{u_N}{\phi} u_N : \phi \in V_N\} =: W_N,
\]
where now $\mathcal P_N$ is the $L^2(\mu)$ projection onto $W_N$.
In practice, we solve~\eqref{eq:saddle_point_formutation}.
As in~\cite{roussel2018spectral},
we choose $V_N$ to be the subspace of $\lp{2}{\mu}$ spanned by the orthonormal basis of functions
\begin{equation}
  \label{eq:basis_functions}
  e_{i,j} = Z^{1/2} \, \e^{\frac{\beta}{2} H(q,p)}
  \, g_i(q) \, h_j(p), \qquad 0 \leq i,j \leq N,
\end{equation}
where $g_i$ are trigonometric functions,
\begin{equation}
  \label{eq:definition_trigonometric_functions}
  g_i(q) =
  \left\{ \begin{aligned}
    (2 \pi)^{-1/2}, \quad & \text{if}~i = 0, \\
    \pi^{-1/2} \sin\left(\frac{i + 1}{2}q\right), \quad & \text{if}~i~\text{is odd}, \\
    \pi^{-1/2} \cos\left(\frac{i}{2}q\right), \quad & \text{if}~i~\text{is even}, i \geq 2, \\
  \end{aligned} \right.
\end{equation}
and $h_j$ are rescaled normalized Hermite functions,
\begin{equation}
  \label{eq:definition_hermite_functions}
  h_j(p) = \frac{1}{\sqrt{\sigma}} \, \psi_j \left( \frac{p}{\sigma} \right),
  \qquad \psi_j (p) := (2 \pi)^{-\frac{1}{4}} \frac{(-1)^j}{\sqrt{j!}} \e^{\frac{p^2}{4}} \, \derivative*{j}{p^j} \, \Bigl( \e^{- \frac{p^2}{2}} \Bigr).
\end{equation}
The functions $(h_j)_{j\in \nat}$ are orthonormal in $\lp{2}{\real}$ regardless of the value of $\sigma$,
which is a scaling parameter that can be adjusted in order to better resolve $\Psi_N$.

As mentioned in \cref{remark:cost_control_variate},
calculating realizations of the control variate requires to evaluate the derivative $\partial_p \psi(q_t, p_t)$ along full paths of the solution to Langevin dynamics.
If~$\psi = \Psi_N$ and its derivative $\partial_p \psi = \partial_p \Psi_N$ are stored as arrays of Fourier/Hermite coefficients,
then the evaluation of $\partial_p \Psi_N$ at a configuration $(q, p)$ is computationally expensive
because it requires to evaluate all the basis functions~\eqref{eq:basis_functions} at this configuration.
Therefore, in practice, the approximate solution~$\Psi_N$ and its gradient are discretized, in a preprocessing step,
over a Cartesian grid with vertices
\begin{equation}
    \label{eq:discretization_points}
    (q_i, p_j) = \bigl(- \pi + 2\pi(i/N_q) , - L_p + 2L_p (j/N_p) \bigr), \qquad 0 \leq i \leq N_q - 1, \quad 0 \leq j \leq N_p.
\end{equation}
The domain in the $p$ direction is truncated at $-L_p$ and $L_p$,
and the parameter~$L_p$ is chosen sufficiently large that escaping the domain is unlikely during a simulation of the Langevin dynamics.
From the values of~$\Psi_N$ at the points~\eqref{eq:discretization_points},
a bilinear interpolant $\widehat \Psi_N$ is constructed over the domain $[-\pi, \pi] \times [-L_p, L_p]$ as follows:
\begin{align*}
    \widehat \Psi_N(q, p)
    &= \Psi_N(q_i,p_j) + \frac{\Psi_N(q_{i+1}, p_j) - \Psi_N(q_i, p_j)}{q_{i+1} - q_i} (q - q_i)
    +\frac{\Psi_N(q_{i}, p_{j+1}) - \Psi_N(q_{i}, p_{j})}{p_{j+1} - p_j} (p - p_j) \\
    &\quad + \frac{\Psi_N(q_{i+1}, p_{j+1}) - \Psi_N(q_{i+1}, p_{j}) - \Psi_N(q_{i}, p_{j+1}) + \Psi_N(q_{i}, p_{j})}{(q_{i+1} - q_i)(p_{j+1} - p_j)}  (q-q_i)(p-p_j),
\end{align*}
where $i = \lfloor (q + \pi) / \Delta q \rfloor$ and $j = \lfloor (p+L_p) / \Delta p \rfloor$,
with $\Delta q = 2\pi/N_q$ and $\Delta p = 2L_p/N_p$.
Here we use the convention that $q_{N_q} = q_{0}$ in view of the $2\pi$ periodicity of $\Psi_N$ in direction $q$.
We emphasize that,
since accessing an array element is an operation with time complexity~$\mathcal O(1)$ with respect to the length of the array,
the computational cost of evaluating~$\widehat \Psi_N$ at a point $(q,p)$ is independent of $N_q$ and $N_p$.

From the bilinear interpolant~$\widehat \Psi_N$,
the control variate~$\xi_t$ is constructed by discretizing~\eqref{eq:definition_control_variate} using the approach presented in~\cref{sub:numerical_results}.
The approximate effective diffusion coefficient~$d[\psi]$,
which enters in the definition~\eqref{eq:improved_estimator} of the estimator~$v(T)$,
is calculated based on~$\widehat \Psi_N$ by numerical quadrature.
The parameters employed for the construction of the control variate described in this section are summarized in~\cref{table:parameters_employed_for_the_construction_of_the_control_variate}.
\begin{table}[ht]
    \centering
    \begin{tabular}{|c|c|c|}
        \hline
        Scaling coefficient of Hermite functions & $\sigma$ & $0.1/\sqrt{\beta}$ \\
        \hline
        \# Fourier/Hermite modes in $q$ or $p$ & $N$ & 300 \\
        \hline
        \# discretization points in $q$ & $N_q$ & 300 \\
        \hline
        \# discretization points in $p$ & $N_p$ & 500 \\
        \hline
        Truncation size of domain & $L_p$ & $9/\sqrt{\beta}$\\
        \hline
    \end{tabular}
    \caption{Parameters employed for the construction of the control variate.}
    \label{table:parameters_employed_for_the_construction_of_the_control_variate}
\end{table}

\subsection{Control variate for the underdamped limit}%
\label{sub:underdamped_approach}
In dimension one, the underdamped limit of the Langevin dynamics is well understood.
Specifically, it is possible to show that the solution to the Poisson equation~\eqref{eq:poisson_equation},
when multiplied by $\gamma$,
converges as $\gamma \to 0$ to a limit $\phi_0$ in $L^2(\mu)$ which can be calculated simply using one-dimensional numerical integration;
see~\cite[Lemma 3.4]{MR2394704} and~\cite[Proposition 4.1]{roussel_thesis}
for proofs of the convergence to $\phi_0$ in $L^2(\mu)$ using probabilistic and analytical arguments,
respectively. See also~\cite{MR2427108} for an explicit calculation of $\phi_0$ in the case where $V$ is a simple cosine potential and for numerical experiments using Monte Carlo simulations as well as the Hermite spectral method. The limiting solution reads~$\phi_0(q,p) = \sign(p) \, \varphi_0\bigl(H(q,p)\bigr)$,
where
\[
    \varphi_0(E) = 2 \pi \int_{E_0}^{\max\{E_0, E\}} \frac{1}{S_{\rm und}(\mathcal E)} \, \d \mathcal E,
    \quad S_{\rm und}(\mathcal E) = \int_{-\pi}^{\pi} P(q, \mathcal E) \, \d q, \quad P(q, \mathcal E) = \sqrt {2 \bigl(\mathcal E - V(q)\bigr)},
\]
and $E_0 = \min_{q \in \torus} V(q)$.
In particular $\phi_0(q,p) = 0$ if $H(q,p) \leq E_0$,
and
\[
    \partial_p \phi_0(q, p)
    = \sign(p) \, \varphi_0'\bigl(H(q,p)\bigr) \, \partial_p H(q,p)
    =
    \begin{cases}
        0 & \text{ if $H(q,p) < E_0$}, \\
        \displaystyle \frac{2 \pi \abs{p}}{S_{\rm und}\bigl(H(q,p)\bigr)} & \text{ if $H(q,p) > E_0$}.
    \end{cases}
\]
For certain potentials,
the function $S_{\rm und}$ can be calculated explicitly,
and in general this function can be approximated accurately by numerical quadrature.
For the potential $V(q) = \frac{1}{2} \bigl(1 - \cos(q)\bigr)$ considered in~\cref{sub:numerical_results},
for example, $S_{\rm und}$ admits an explicit expression in terms of an elliptic integral of the second kind~\cite{MR2427108}.
In practice, given an explicit expression or implementation of~$S_{\rm und}$,
we calculate~$\varphi_0$ over a large interval $[0, E_{\max}]$ with $E_{\max} = 100 \beta$
using the ODE solver from the Julia module \texttt{DifferentialEquations.jl}~\cite{rackauckas2017differentialequations} with default parameters.
This returns an object containing an approximation of $\varphi_0$ at discrete values of $E$ automatically selected in order to meet a default accuracy threshold.
Conveniently, this object can be evaluated at any~$E \in [0, E_{\max}]$,
in which case an approximation $\varphi_0(E)$ automatically obtained by a high-order interpolation from the discrete solution is returned.

\begin{remark}
    In practice,
    several control variates can be calculated simultaneously during the simulation,
    and only the one leading to the smallest variance can be retained in the estimator.
    Alternatively, the control variates can be combined in order to minimize the variance;
    specifically, given two approximations $\xi^{(1)}_T$ and $\xi^{(2)}_T$ of $q_T - q_0$,
    one may consider the following composite estimator instead of \eqref{eq:improved_estimator}:
    \[
        \widehat v(T) =  \frac{\bigl\lvert q_T - q_0 \bigr\rvert^2}{2T}
        - \alpha_1 \left( \frac{\bigl\lvert \xi^{(1)}_T \bigr\rvert^2}{2T} + \lim_{T \to \infty} \expect \left[\frac{\bigl\lvert \xi^{(1)}_T \bigr\rvert^2}{2T}\right] \right)
        - \alpha_2 \left( \frac{\bigl\lvert \xi^{(2)}_T \bigr\rvert^2}{2T} + \lim_{T \to \infty} \expect \left[\frac{\bigl\lvert \xi^{(2)}_T \bigr\rvert^2}{2T}\right] \right).
    \]
    There are systematic \emph{a posteriori} approaches for choosing optimally $\alpha_1$ and $\alpha_2$ in order to minimize the variance,
    which are studied, for example, in~\cite{MR1958845}.
    For the sake of simplicity, we do not implement or study these approaches here.
\end{remark}

\subsection{Numerical results}%
\label{sub:numerical_results}
We discretize the Langevin dynamics using the geometric Langevin algorithm introduced in~\cite{MR2608370} which,
in the general multi-dimensional case, is based on the iteration
\begin{align*}
    \vect p^{n+1/2}
    &= \vect p^n - \frac{\Delta t}{2} \grad V(\vect q^n), \\
    \vect q^{n+1/2}
    &= \vect q^n + \Delta t \, \vect p^{n+1/2}, \\
    \widetilde {\vect p}^{n+1}
    &= \vect p^{n+1/2} - \frac{\Delta t}{2} \grad V(\vect q^{n+1}), \\
    \vect p^{n+1} &= \exp \left(- \gamma \Delta t \right) \widetilde {\vect p}^{n+1}
    + \sqrt{2 \gamma \beta^{-1}} \vect g^n, \qquad \vect g^n = \int_{n \Delta t}^{(n+1)\Delta t} \e^{-\gamma \bigl((n+1)\Delta t-s\bigr)} \, \d \vect w_s,
\end{align*}
supplemented with the initial condition $(\vect q^0, \vect p^0) \sim \mu$.
The resulting discrete-time process $(\vect q^n, \vect p^n)_{n \in \nat}$ is an approximation of $(\vect q_{n \Delta t}, \vect p_{n \Delta t})_{n \in \nat}$.
The first three lines can be viewed as a Strang splitting of the Hamiltonian part of the dynamics;
this is the St\o rmer--Verlet scheme~\cite{verlet1967computer}.
The fourth line is an analytical integration of the fluctuation/dissipation part.
We write the terms~$(\vect g^n)_{n\geq0}$ as stochastic integrals, instead of giving only their law,
because these are correlated with the Brownian increments necessary for constructing the control variate.
Specifically, the It\^o integral in the definition~\eqref{eq:definition_control_variate} of $\xi$ is approximated using the explicit scheme
\[
    I_{\psi}^{n+1} = I_{\psi}^n + \sqrt{2 \gamma \beta^{-1}} \, \grad_{\vect p} \psi(\vect q^{n}, \vect p^{n}) \cdot
    \widetilde {\vect g}^n, \qquad \widetilde {\vect g}^n = \vect w_{(n+1)\Delta t} - \vect w_{n \Delta t},
\]
where $(\widetilde {\vect g}^n)_{n \geq 0}$ are $d$-dimensional standard Gaussian random variables correlated with $(\vect g^n)_{n \geq 0}$
and $I_{\psi}^{n}$ is an approximation of $\sqrt{2 \gamma \beta^{-1}} \int_{0}^{n \Delta t} \grad_{\vect p} \psi_{\vect e}(q_t, p_t) \cdot \d \vect w_t$.
An explicit calculation using It\^o's isometry shows that the vectors $(\vect g^n, \widetilde {\vect g}^n)_{n \geq 0}$,
which are independent and identically distributed for different values of $n$,
are normally distributed with mean 0 and covariance matrix
\[
    \mat S =
    \begin{pmatrix}
        \displaystyle \frac{1}{2 \gamma} \bigl(1 - \e^{-2 \gamma \Delta t}\bigr) \mat I_d
        & \displaystyle \frac{1}{\gamma}\bigl(1 - \e^{-\gamma \Delta t}\bigr) \mat I_d \\[.4cm]
        \displaystyle  \frac{1}{\gamma}\bigl(1 - \e^{-\gamma \Delta t}\bigr) \mat I_d
        & \displaystyle \Delta t \, \mat I_d
    \end{pmatrix}.
\]
In practice, we generate $(\vect g^n, \widetilde {\vect g}^n)$ as
\[
    \begin{pmatrix}
        \vect g^n \\
        \widetilde {\vect g}^n
    \end{pmatrix}
    =
    \sqrt{\mat S}
    \begin{pmatrix}
        \mathcal  G^n \\
        \widetilde {\mathcal  G}^n
    \end{pmatrix},
    \qquad
    \begin{pmatrix}
        \mathcal  G^n \\
        \widetilde {\mathcal  G}^n
    \end{pmatrix}
    \overset{\rm iid} \sim
    \mathcal N(0, \mat I_{2d}),
\]
where $\sqrt{\mat S}$ is the unique symmetric, positive definite matrix square root of $\mat S$.


For a given final time $T$,
the expectations and standard deviation of the estimators~$u(T)$ and~$v(T)$,
defined respectively in~\eqref{eq:simple_estimator} and~\eqref{eq:improved_estimator},
are estimated from a number $J$ of realizations.
The parameters employed in the simulation are summarized in~\cref{table:parameters_employed_for_mc},
and the associated numerical results are presented in~\cref{fig:effective_diffusion_langevin}.
\begin{table}[ht]
    \centering
    \begin{tabular}{|c|c|c|}
        \hline
        Time step & $\Delta t$ & 0.01 \\
        \hline
        \# Number of realizations & $J$ & 5000 \\
        \hline
        \# Final time & $T$ & $100 \gamma^{-1}$ \\
        \hline
    \end{tabular}
    \caption{Parameters employed for the Monte Carlo simulation.}
    \label{table:parameters_employed_for_mc}
\end{table}
We observe that the sample means corresponding to each of the estimators are in good agreement,
and that for $\gamma = 10^{-5}$ the effective diffusion coefficient is very close, in relative terms,
to its theoretical limit $D_{\rm und}/\gamma$.
The Galerkin method for the Poisson equation~\eqref{eq:poisson_equation} is inaccurate for $\gamma \ll 1$,
which explains the mismatch in this regime between the curve labeled ``Galerkin'',
which corresponds to a deterministic approximation of the effective diffusion coefficient from the numerical solution to the Poisson equation,
and the other curves in the left-panel.

The two control variate approaches yield computational benefits in different regimes:
the control variate constructed from the Fourier/Hermite approximation of the solution to Poisson equation~\eqref{eq:poisson_equation}
enables a variance reduction by a factor more than 100 for $\gamma \geq 10^{-2}$,
but this factor decreases as $\gamma \to 0$.
This is not surprising since,
if the basis~\eqref{eq:basis_functions} is fixed with respect to~$\gamma$,
then the accuracy of the Galerkin method~\eqref{eq:saddle_point_formutation} becomes worse and worse as $\gamma \to 0$;
see~\cite{roussel2018spectral}.
In contrast, the control variate constructed using the ``underdamped'' strategy of \cref{sub:underdamped_approach} enables a variance reduction by a factor close to 100 for the smallest value of $\gamma$ considered (namely $10^{-5}$),
but the benefits decrease as $\gamma$ increases.

\begin{figure}[ht]
    \centering
    \includegraphics[width=0.99\linewidth]{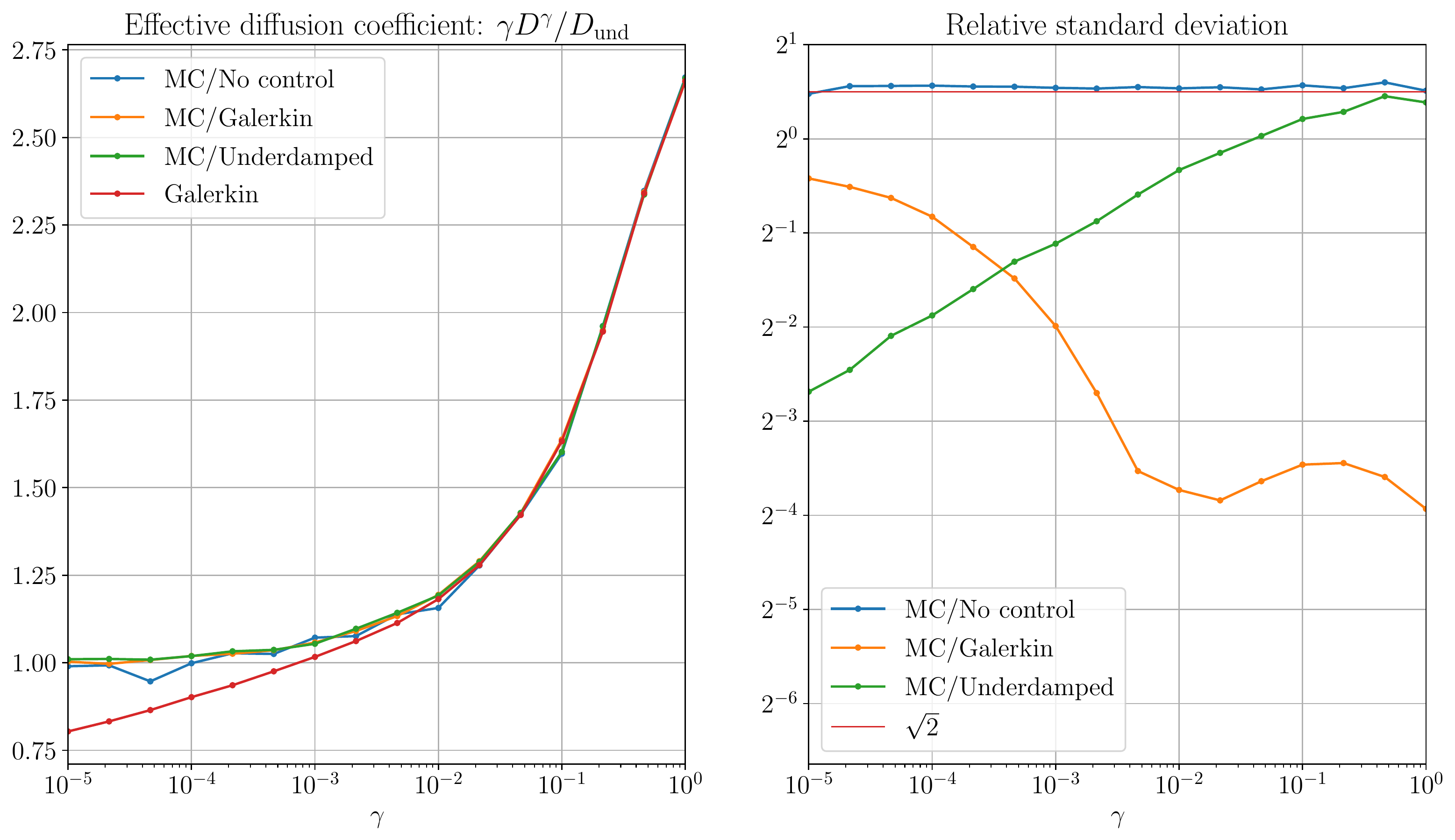}
    \caption{
        Effective diffusion coefficient and relative standard deviation $\frac{\sqrt{\var\bigl[v(T)\bigr]}} {D^{\gamma}}$ of the estimators considered at time $T = 100/\gamma$.
        The data labeled ``MC/No control'' correspond to Monte Carlo simulations without a control variate,
        i.e.\ to the estimator $u(T)$ given in~\eqref{eq:simple_estimator}.
        The data labeled ``MC/Galerkin'' and ``MC/Underdamped'' correspond to the improved estimator~\eqref{eq:definition_control_variate},
        with $\psi$ obtained using the approaches of~\cref{sub:galerkin_approach,sub:underdamped_approach},
        respectively.
        Finally, the curve labeled ``Galerkin'' is the approximate diffusion coefficient obtained by the Galerkin method alone,
        which is given by~$\ip*{\widehat \Psi_N}{p}$ in the notation of \cref{sub:galerkin_approach}.
        The value $\sqrt{2}$ is the asymptotic relative standard deviation for the simple estimator $u(T)$,
        see~\cref{remark:asym_variance_u}.
    }%
    \label{fig:effective_diffusion_langevin}
\end{figure}

\Cref{fig:time_bias_variance} illustrates the evolution of the expectation and standard deviation of the estimators~$u(t)$ given in~\eqref{eq:simple_estimator} and~$v(t)$ given in~\eqref{eq:improved_estimator},
empirically estimated from $J = 5000$ realizations,
with respect to the integration time~$t$.
It appears clearly that,
for the value $\gamma = 10^{-3}$ considered,
the improved estimators $v(t)$ obtained using the approaches outlined in \cref{sub:galerkin_approach,sub:underdamped_approach}
have a much smaller variance than~$u(t)$ throughout the simulation.
\begin{figure}[ht]
    \centering
    \includegraphics[width=0.99\linewidth]{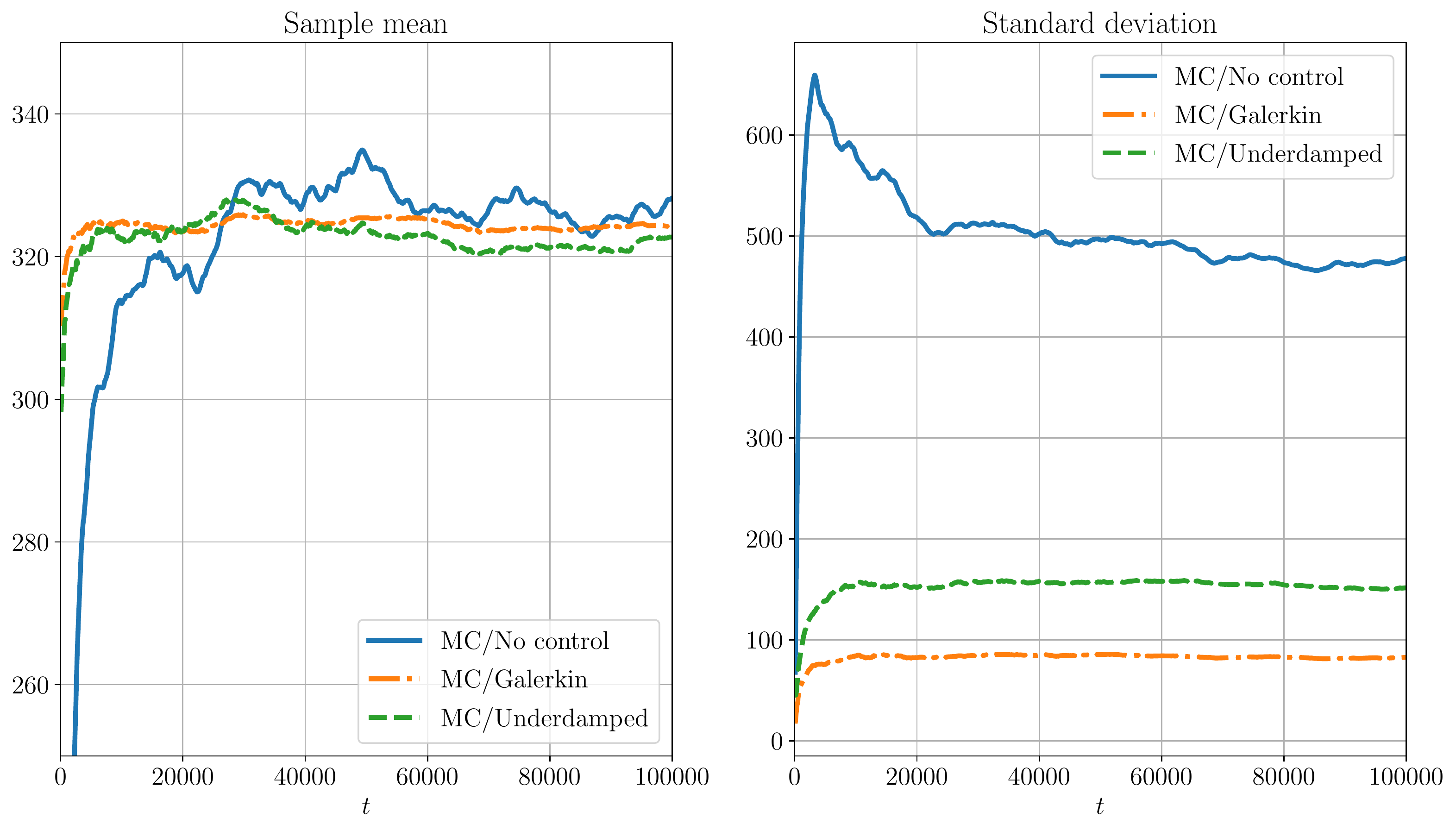}
    \caption{
        Sample mean and sample standard deviation of the estimators $u(t)$ in~\eqref{eq:simple_estimator} and $v(t)$ in~\eqref{eq:improved_estimator},
        for the friction parameter $\gamma = 10^{-3}$.
        The approximate solution to the Poisson equation used for constructing the control variate appearing in $v(t)$ here
        is that given in \cref{sub:underdamped_approach}.
    }%
    \label{fig:time_bias_variance}
\end{figure}

\subsection{Extension to the generalized Langevin dynamics}%
\label{sub:generalization_to_generalized_langevin_dynamics}
The variance reduction approach described in \cref{sec:method},
in particular with the control variate constructed from the limiting solution to the Poisson equation as $\gamma \to 0$,
may be extended for calculating the mobility of simple generalized Langevin dynamics in one spatial dimension.
The paradigmatic example dynamics we consider here is the following,
which is studied in~\cite{MR2793823,GPGSUV21}:
\begin{equation}
\label{eq:gle}
\left\{
  \begin{aligned}
      & \d q_t = p_t \, \d t, \\
      & \d p_t = - V'(q_t) \, \d t + \frac{\sqrt{\gamma}}{\nu} \, z_t \, \d t, \\
      & \d z_t = - \frac{\sqrt{\gamma}}{\nu} \, p_t  \, \d t
      -   \frac{1}{\nu^2} \, z_t \, \d t + \sqrt{\frac{2 \beta^{-1}}{\nu^2}} \, \d w_t,
  \end{aligned}
\right.
\end{equation}
where $z_t \in \real$.
The unique invariant probability measure for this dynamics is given by
\[
    \mu_{\rm GLE}(\d q \, \d p \, \d z) \propto \exp \biggl( - \beta \left( H(q,p) + \frac{z^2}{2} \right) \biggr) \, \d q \, \d p \, \d z.
\]
As proved in~\cite{ottobre2012asymptotic},
a functional central limit theorem applies also to the dynamics~\eqref{eq:gle}:
the diffusively rescaled position process $(\varepsilon q_{t/\varepsilon^2})_{t \geq 0}$ converges in distribution,
in the Banach space of continuous functions over bounded time intervals,
to a Brownian motion with diffusion coefficient $ D^{\gamma, \nu}$.
As in the case of the underdamped Langevin dynamics, this diffusion coefficient can be calculated in terms of the solution of an appropriate Poisson equation: $  D^{\gamma, \nu} = \ip{\phi}{p}$, where $\ip{\dummy}{\dummy}$ is the inner product of $L^2(\mu_{\rm GLE})$
and $\phi$ is the unique solution to the following Poisson equation posed in $L^2_0(\mu_{\rm GLE})$:
\begin{equation}
    \label{eq:poisson_gle}
    - \mathcal L_{\rm GLE} \phi = p,
\end{equation}
where $\mathcal L_{\rm GLE}$ is the generator of~\eqref{eq:gle}.
As for the underdamped Langevin dynamics, the (Markovian approximation of the) generalized Langevin dynamics is difficult to understand in the underdamped regime,
and in particular there does not exist a rigorous result on the behavior of~$D^{\gamma, \nu}$ in the limit as $\gamma \to 0$.
Our goal in this section is to calculate accurately the mobility for the dynamics~\eqref{eq:gle} in the underdamped regime
using a control variate approach similar to that described in~\cref{sec:method},
and to assess in this manner the validity of the asymptotic scaling for~$D^{\gamma,\nu}$ conjectured in~\cite{GPGSUV21} by means of formal asymptotics.
An application of It\^o's formula gives
\[
    q_T - q_0 = \int_{0}^{T} p_t \, \d t
    = \phi(q_0, p_0, z_0) - \phi(q_T, p_T, z_T) + \sqrt{\frac{2 \beta^{-1}}{\nu^2}} \int_{0}^{T} \partial_z \phi(q_t, p_t, z_t) \, \d w_t,
\]
where $\phi$ is now the solution to~\eqref{eq:poisson_gle}.
This motivates the following estimator for the mobility:
\begin{subequations}
\begin{equation}
    \label{eq:improved_estimator_gle}
    v(T) = d[\psi] + \frac{1}{2T} \left( \abs*{q_T - q_0}^2 - \abs{\xi_T}^2\right),
    \qquad d[\psi] := \frac{\beta^{-1}}{\nu^2} \int \abs{\partial_z \psi}^2 \, \d \mu_{\rm GLE},
\end{equation}
where
\begin{align}
    \label{eq:definition_control_variate_gle}
    \xi_T = \psi(q_0, p_0, z_0) - \psi(q_T, p_T, z_T) + \sqrt{\frac{2 \beta^{-1}}{\nu^2}} \int_{0}^{T} \partial_z \psi(q_t, p_t, z_t) \, \d w_t,
\end{align}
\end{subequations}
and $\psi$ is an approximate solution to the Poisson equation~\eqref{eq:poisson_gle}.
The initial condition~$(q_0, p_0, z_0)$ is distributed according the invariant measure of the process,
i.e.~$\mu_{\rm GLE}$.

In~\cite{GPGSUV21},
we employed an asymptotic expansion of the form
\(
    \phi = \gamma^{-1} \phi_0 + \gamma^{-1/2} \phi_1 + \gamma^{0} \phi_2 + \dotsb
\)
in order to study the underdamped limit
and we derived expressions for $\phi_0$ and $\phi_1$
which enable to show formally that $D^{\gamma, \nu}$ behaves as $1/\gamma$ in the limit as $\gamma \to 0$,
with a prefactor $D^{\rm und}_{\nu}$ that can be efficiently calculated and is different from $D^{\rm und}$.
(Recall from~\cref{sec:introduction} that the diffusion coefficient $D^{\rm und}$ is defined as~$D^{\rm und} = \lim_{\gamma \to 0} \gamma D^{\gamma}$,
where $D^{\gamma}$ is the effective diffusion coefficient of the one-dimensional Langevin dynamics~\eqref{eq:effective_diffusion_poisson}.)
Although the assumed asymptotic expansion is shown to be invalid in~\cite{GPGSUV21}
because $\mathcal L_{\rm GLE} \phi_1 \notin \lp{2}{\mu}$,
our numerical results in this section demonstrate that this expansion can still be leveraged for constructing an efficient control variate $\xi_T$ in~\eqref{eq:definition_control_variate_gle}.
Specifically,
we obtain a considerable reduction in variance by choosing the approximation $\psi$ in~\eqref{eq:definition_control_variate_gle} as $\psi = \gamma^{-1} \phi_0 + \gamma^{-1/2} \phi_1$.
We refer to~\cite[Section~4.3.2] {GPGSUV21} for the expressions of the asymptotic value $D^{\nu}_{\rm und}$ and of the functions $\phi_0$ and $\phi_1$.

For the numerical integration of~\eqref{eq:gle},
we employ a numerical scheme similar to that presented in~\cref{sub:numerical_results} for the Langevin dynamics.
The scheme we use may be abbreviated as BABO using the terminology of~\cite{MR4379630},
although it is not explicitly studied in that paper.
More precisely, compared to the scheme used in~\cref{sub:numerical_results},
the last update for $p$ is replaced by the following simultaneous update for $p$ and $z$,
corresponding to the analytical integration of the Ornstein--Uhlenbeck part of the dynamics:
\begin{equation*}
    \begin{pmatrix}
        p^{n+1} \\
        z^{n+1}
    \end{pmatrix}
    =
    \e^{\mat M \Delta t}
    \begin{pmatrix}
        p^{n} \\
        z^{n}
    \end{pmatrix}
    + \int_{n \Delta t}^{(n+1)\Delta t} \e^{\mat M \bigl((n+1) \Delta t - s\bigr)}
    \begin{pmatrix}
        0 \\
        \sqrt{\frac{2 \beta^{-1}}{\nu^2}}
    \end{pmatrix}
    \, \d w_s,
    \qquad \mat M =
    \begin{pmatrix}
        0 & \frac{\sqrt{\gamma}}{\nu} \\
        - \frac{\sqrt{\gamma}}{\nu} & -   \frac{1}{\nu^2}
    \end{pmatrix}.
\end{equation*}
The It\^o integral in this equation, which we denote by $\vect I_n$,
is a bivariate Gaussian with mean~0 and a covariance matrix independent of~$n$.
The covariance matrix is calculated from It\^o's isometry only once, at the beginning of the simulation.
Likewise, the matrix exponential $\e^{\mat M \Delta t}$ is precalculated before simulating the GLE dynamics.

The It\^o integral in~\eqref{eq:definition_control_variate_gle} is discretized by using the Euler--Maruyama method.
Since the Brownian increment $w_{(n+1)\Delta t} - w_{n \Delta t}$ is correlated to the It\^o integral~$\vect I_n$,
a technique similar to that presented~\cref{sub:numerical_results} is required in order to generate $(\vect I_n, w_{(n+1)\Delta t} - w_{n \Delta t})$ at each iteration.

\Cref{fig:effective_diffusion_time_gle} illustrates the evolution of $\expect \bigl[ u(t) \bigr]$ and $\expect \bigl[ v(t) \bigr]$,
given in~\eqref{eq:simple_estimator} and~\eqref{eq:improved_estimator_gle},
with respect to time,
for a value of~$\nu = 2$ that is sufficiently large to observe a different asymptotic behavior in the limit as $\gamma \to 0$ than that of standard Langevin dynamics.
These expectations are estimated from simulations using the same parameters as in~\cref{table:parameters_employed_for_mc}.
The associated $[m - 3 s, m + 3 s]$ confidence intervals (corresponding to a confidence of approximately $99.7\%$ assuming Gaussianity) are also depicted,
where $m$ and $s$ are the sample mean and sample standard deviation.
It is evident from the figures that the control variate leads to considerable improvements in terms of variance.
Drawing conclusions on the bias is a more delicate task,
as the true value of $D^{\gamma, \nu}$ is unknown,
but it is clear that the improved estimator~\eqref{eq:improved_estimator_gle} has a smaller bias for small times.
Furthermore, we observe that the effective diffusion coefficient for $\gamma = 10^{-5}$
is in very good agreement with the asymptotic equivalent $D^{\rm und}_{\nu}/\gamma$,
where $D^{\rm und}_{\nu} = \lim_{\gamma \to 0} \gamma D^{\gamma, \nu}$ is the limiting value conjectured in~\cite{GPGSUV21}.
The dependence of the effective diffusion coefficient on~$\gamma$,
calculated using the control variate, is presented in~\cref{fig:effective_diffusion_gle}.
\begin{figure}[ht]
    \centering
    \includegraphics[width=0.8\linewidth]{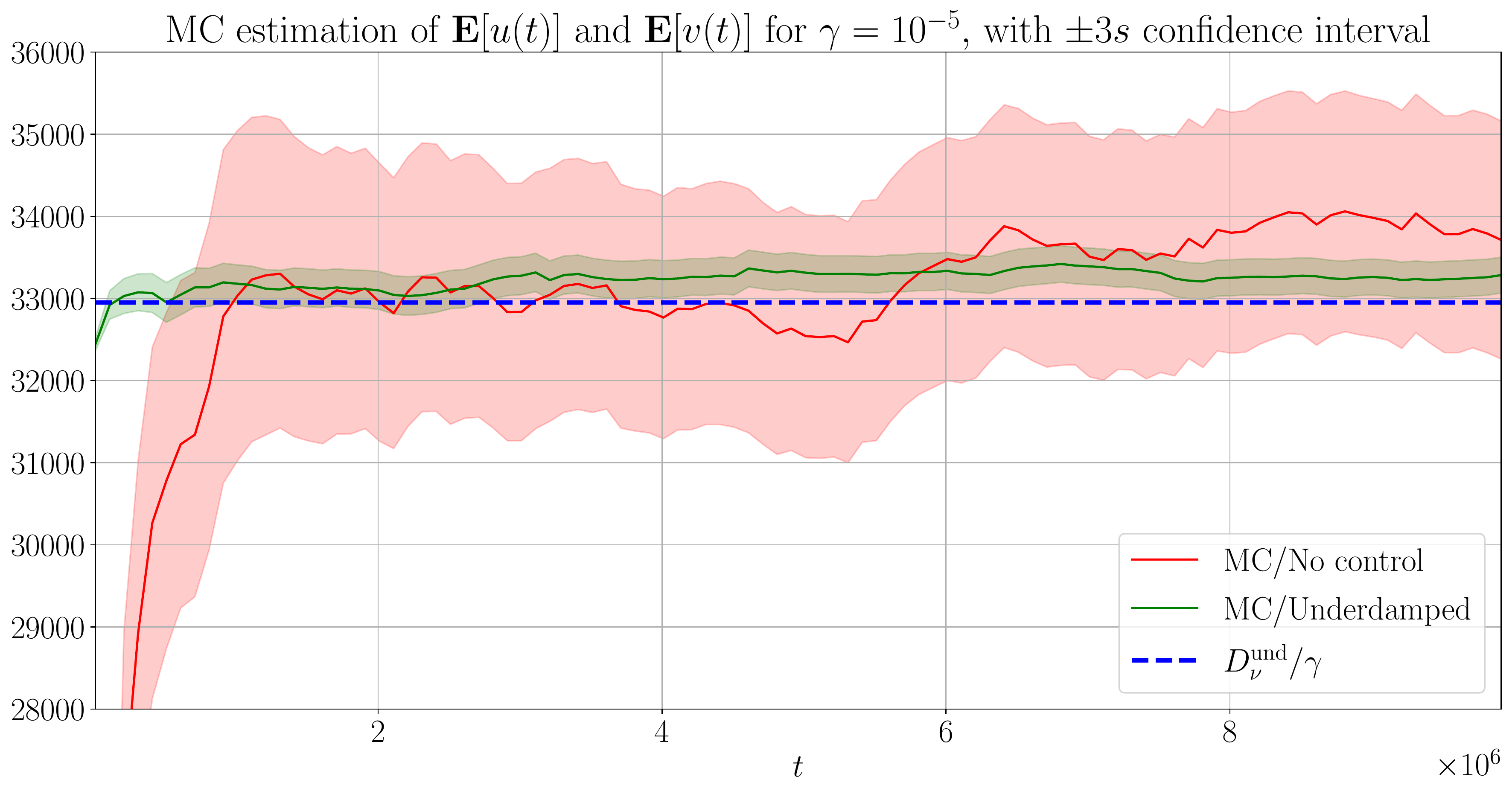}
    \caption{%
        Expectations of the naive~\eqref{eq:simple_estimator} and improved~\eqref{eq:improved_estimator_gle} estimators for generalized Langevin dynamics
        and associated ``$m \pm 3 s$'' confidence intervals,
        estimated from $J = 5000$ realizations.
    }
    \label{fig:effective_diffusion_time_gle}
\end{figure}
\begin{figure}[ht]
    \centering
    \includegraphics[width=0.9\linewidth]{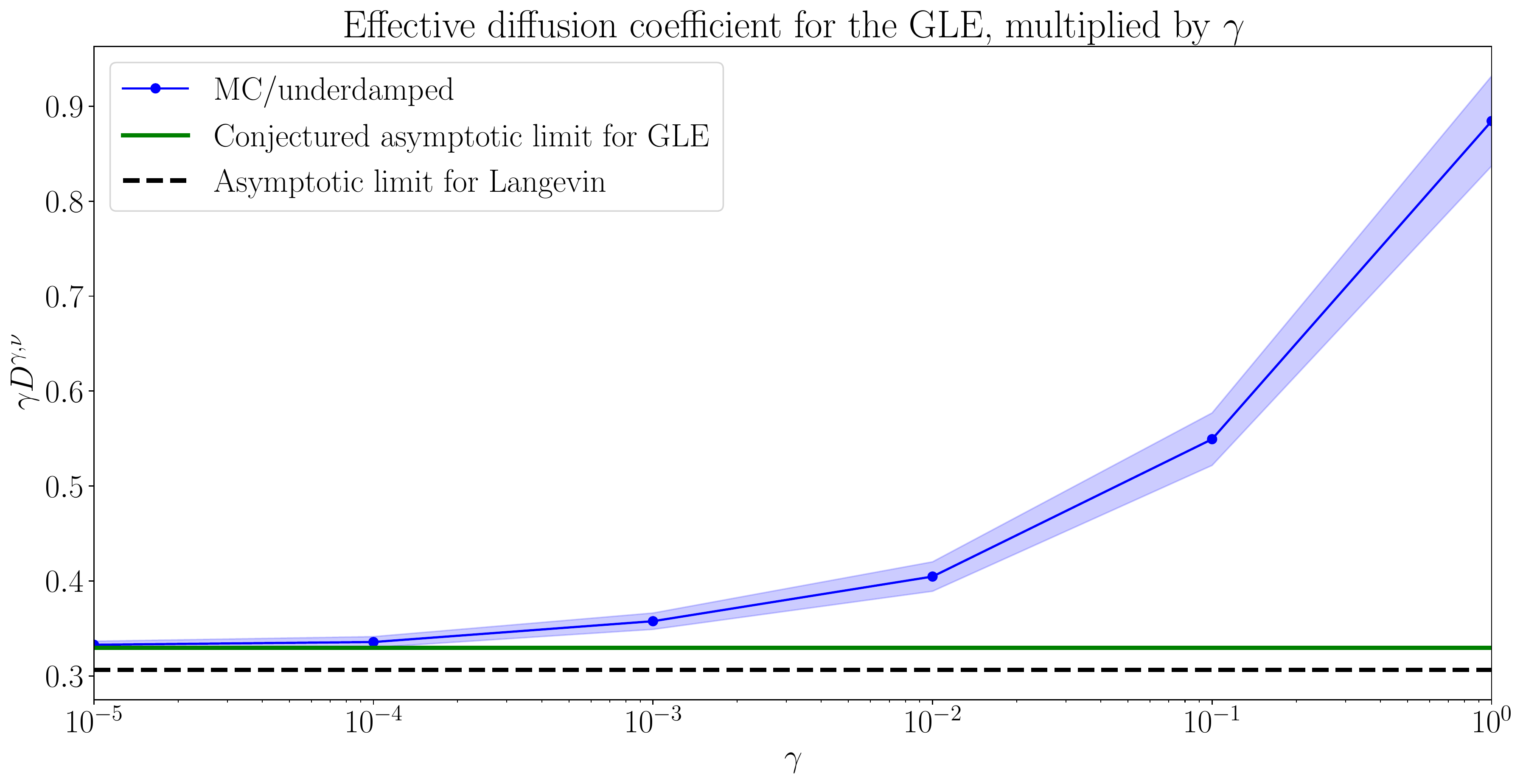}
    \includegraphics[width=0.9\linewidth]{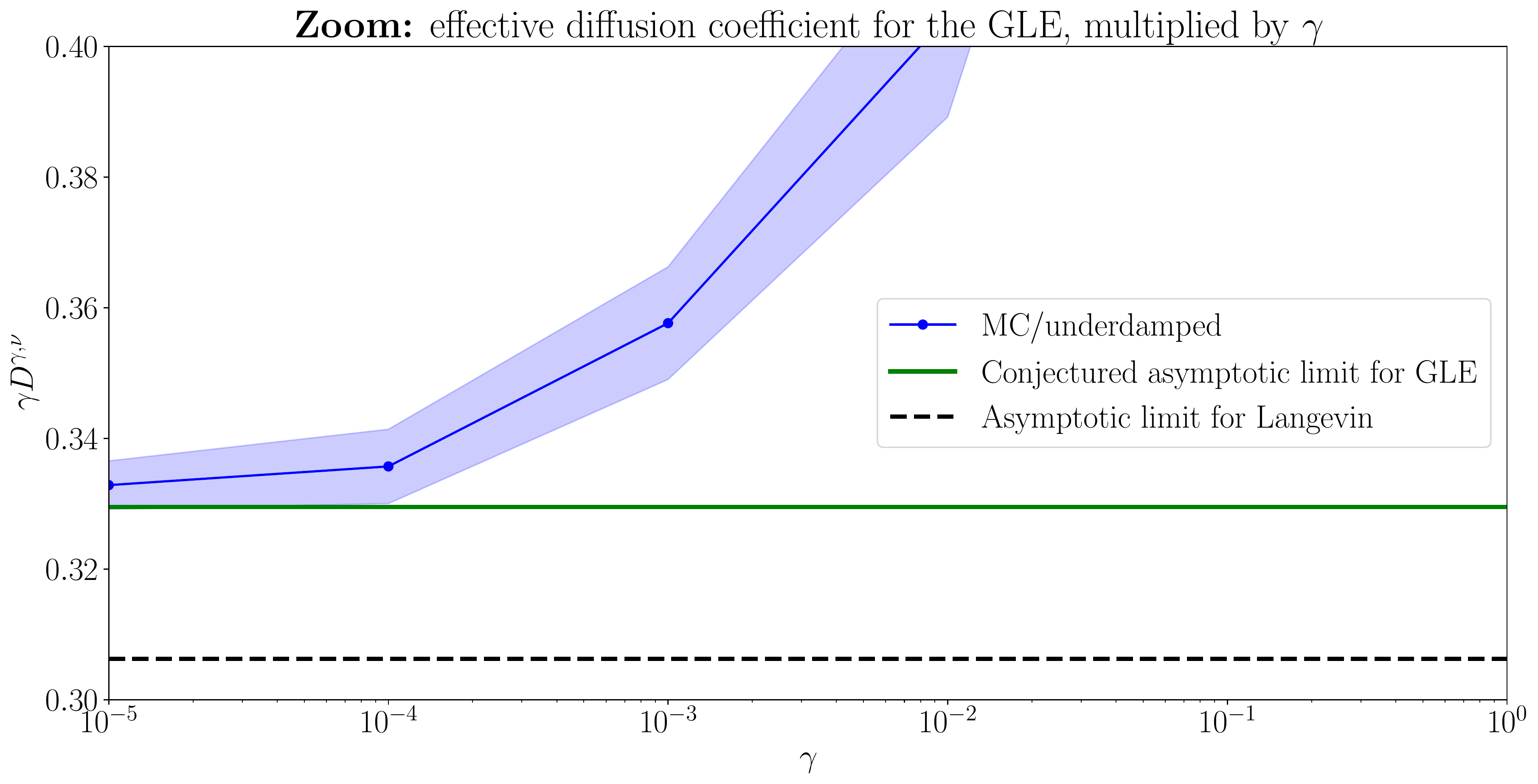}
    \caption{%
        Expectation and ``$m \pm 3 s$'' confidence intervals for $\gamma \expect [v(T)]$ in the underdamped limit,
        for generalized Langevin dynamics with $\nu = 2$.
        Since $T$ scales as $1/\gamma$ with a large prefactor,
        it is expected that $\expect \bigl[ v(T) \bigr] \approx D^{\gamma,\nu}$.
    }
    \label{fig:effective_diffusion_gle}
\end{figure}


\section{Application to Langevin dynamics in two dimensions}%
\label{sec:applications_2d}%
The approaches employed in \cref{sec:application_to_one_dimensional_langevin_type_dynamics} for constructing an approximate solution to the Poisson equation~\eqref{eq:poisson_equation}
do not generalize well to the multi-dimensional setting for non-separable potentials.
On one hand, Galerkin methods for the Poisson equation suffer from the curse of dimensionality and,
on the other hand, the behavior of the solution to the Poisson equation is not well understood in the underdamped limit.
In this section, we discuss alternative approaches.
We begin by showing that,
under symmetry assumptions on the potential~$V$,
the diffusion tensor $\mat D^{\gamma}$ is isotropic.

\begin{lemma}
    \label{lemma:diffusion_symmmetry}
    If $V$ satisfies the symmetry relation $V(q_1, q_2) = V(q_2, q_1)$,
    then $D^{\gamma}_{11} = D^{\gamma}_{22}$.
    In addition, if $V(q_1, q_2) = V(q_1, -q_2)$ or $V(q_1, q_2) = V(-q_1, q_2)$,
    then $D_{12} = 0$.
\end{lemma}
\begin{proof}
    The first claim is obvious by symmetry.
    Here we prove only that $D^{\gamma}_{12} = 0$ if the potential satisfies $V(q_1, q_2) = V(q_1, -q_2)$.
    To this end,
    let $\mathcal R$ be the operator on~$L^2_0(\mu)$ defined by $\mathcal R f(q_1, p_1, q_2, p_2) = f(q_1, p_1, -q_2, -p_2)$.
    A simple calculation shows that~$\mathcal R \mathcal L \mathcal R = \mathcal L$ and so,
    denoting by $\phi_1$ the solution to $- \mathcal L \phi_1 = p_1$ posed on~$L^2_0(\mu)$,
    we have $- \mathcal R \mathcal L \mathcal R \phi_1 = p_1$.
    This implies  $- \mathcal L \mathcal R \phi_1 = p_1$ because $\mathcal R^2 = \id$,
    so $\phi_1 = \mathcal R \phi_1$ by uniqueness of the solution to the Poisson equation.
    It then follows that $D^{\gamma}_{12} = \ip{\phi_1}{p_2} = 0$.
\end{proof}

We consider in this work a non-separable potential even simpler than~\eqref{eq:potential_julien}:
\begin{align}
    \label{eq:potential_simple}
    V(q) =  \mathcal V(q_1) + \mathcal V(q_2) + \delta \mathcal W(q_1, q_2) := - \frac{\cos(q_1) + \cos(q_2)}{2} - \delta \cos(q_1) \cos(q_2).
\end{align}
This potential satisfies the symmetry assumption of~\cref{lemma:diffusion_symmmetry},
and so the corresponding diffusion tensor is isotropic.
Therefore, in order to simplify the discussion,
we focus on estimating~$D^{\gamma}_{\vect e}$ only for the unit vector~$\vect e = (1, 0)^\t$.
Accordingly, let $\phi_\delta(q, p)$ denote the solution to the Poisson equation $- \mathcal L^{\delta} \phi_\delta = p_1$,
where the generator~$\mathcal L^{\delta}$ of the dynamics now reads
\begin{align}
    \label{eq:generator}
    \mathcal L^{\delta} = \mathcal L_0 + \delta \mathcal L_1
    := L_1 + L_2
    - \delta \grad \mathcal W \cdot \grad_{\vect p},
\end{align}
with $L_i = p_i \partial_{q_i} - \mathcal V'(q_i) \, \partial_{p_i} + \gamma \left(- p_i \partial_{p_i} + \beta^{-1} \partial^2_{p_i} \right)$
for $i \in \{1, 2\}$.
Notice that $\phi_\delta(\vect q, \vect p) = \phi_{\rm 1D}(q_1, p_1)$ when $\delta = 0$,
with $\phi_{\rm 1D}$ the solution to the one-dimensional Poisson equation $-L_1 \phi_{\rm 1D}(q_1, p_1) = p_1$.
In other words, it holds in this case that $\phi_0 = \phi_{\rm 1D} \otimes 1$,
where for two functions $f_1: \torus \times \real \rightarrow \real$ and $f_2: \torus \times \real \rightarrow \real$
the notation $f_1 \otimes f_2$ denotes the function $(\vect q, \vect p) \mapsto f_1(q_1,p_1) \, f_2(q_2,p_2)$.
For small $\delta$, it is natural to use $\phi_0$, or an approximation thereof,
as the function $\psi_{\vect e}$ in the definition of the control variate~\eqref{eq:definition_control_variate}.
Note that $d[\psi_{\vect e}]$ in~\eqref{eq:definition_control_variate} is an average with respect to the invariant distribution of the non-separable dynamics,
which we denote by $\mu_{\delta}$ to emphasize the dependence on $\delta$.
Specifically, we have
\begin{equation}
    \label{eq:effective_diffusion_d}
    d[\psi_{\vect e}]
    = \int_{\torus^2 \times \real^2} \abs{\grad_{\vect p} \psi_{\vect e}}^2 \, \d \mu_{\delta}.
\end{equation}
The following result establishes that~$\phi_\delta$ converges to $\phi_0$ in $L^2(\mu)$ in the limit as~$\delta \to 0$.

\begin{proposition}
    \label{proposition:convergence_gradient}
    Let $\phi_{\delta}$ and $\phi_0$ denote the solutions to the Poisson equation $- \mathcal L^{\delta} \phi_{\delta} = p_1$
    and its separable counterpart $- \mathcal L_0 \phi_0 = p_1$,
    these equations being posed in $L^2_0(\mu_{\delta})$ and $L^2_0(\mu_0)$, respectively.
    Then there exists $C > 0$ independent of $\delta$ and $\gamma$ such that
    \begin{align}
        \notag
        &\forall (\gamma, \delta) \in (0, \infty) \times [-1, 1], \qquad \\
        \label{eq:convergence_delta}
        &\qquad \norm*{\phi_\delta - \phi_0}[L^2(\mu_{\delta})]
        \leq C\lvert \delta \rvert \left(
            \max\bigl\{\gamma,\gamma^{-1}\bigr\} \norm{\grad_{\vect p} \phi_0}[L^2(\mu_{\delta})]
            + \norm{\phi_0}[L^2(\mu_{\delta})]
        \right),
    \end{align}
    and
    \begin{align}
        \label{eq:convergence_gradient}
        \forall (\gamma, \delta) \in (0, \infty) \times [-1, 1], \qquad
        \norm*{\grad_{\vect p} \phi_\delta - \grad_{\vect p} \phi_0}[L^2(\mu_{\delta})]
        \leq C \lvert \delta \rvert  \max\bigl\{1,\gamma^{-1}\bigr\} \norm{\grad_{\vect p} \phi_0}[L^2(\mu_{\delta})] .
    \end{align}
\end{proposition}
Before we present the proof of this result,
a couple of remarks are in order.
\begin{remark}
    \label{remark:norms_equivalent}
    There exist constants $c_1, c_2 > 0$ such that
    \[
        \forall \delta \in [-1, 1], \qquad
        c_1 \norm{\dummy}[L^2(\mu_{\delta})] \leq \norm{\dummy}[L^2(\mu_0)] \leq c_2 \norm{\dummy}[L^2(\mu_{\delta})],
    \]
    and so the inequalities~\eqref{eq:convergence_delta} and~\eqref{eq:convergence_gradient} are valid also with~$\norm{\dummy}[L^2(\mu_{\delta})]$ replaced by $\norm{\dummy}[L^2(\mu_{0})]$.
\end{remark}
\begin{remark}
Using the definition~\eqref{eq:effective_diffusion_d} and~\cref{proposition:convergence_gradient},
we deduce that
\[
    \forall (\gamma, \delta) \in (0, 1) \times [-1, 1], \qquad
    \frac{ \left\lvert \sqrt{d[\phi_{\delta}]} - \sqrt{d[\phi_{0}]} \right\rvert } {\sqrt{d[\phi_{0}]}}
    \leq C \lvert\delta\rvert \gamma^{-1}.
\]
Consequently,
since if $\abs{\sqrt{x} - 1} \leq \varepsilon$ for $x \geq 0$ then $\abs{x-1} = \abs{\sqrt{x} - 1} \abs{\sqrt{x} + 1} \leq \varepsilon (2 + \varepsilon)$,
\[
    \forall (\gamma, \delta) \in (0, 1) \times [-1, 1], \qquad
    \frac{ \lvert d[\phi_{\delta}] - d[\phi_{0}] \rvert } {d[\phi_{0}]}
    \leq  C \lvert\delta\rvert \gamma^{-1} (2 + C \lvert\delta\rvert \gamma^{-1}).
\]
In particular, for~$\delta/\gamma$ fixed and sufficiently small,
it holds that $\liminf_{\gamma \to 0} \gamma D^{\gamma}_{\vect e} > 0$;
that is, the effective diffusion coefficient scales as $\gamma^{-1}$ in this case.
(Of course, this does not say anything about the behavior of the diffusion coefficient for $\abs{\delta} > 0$ fixed and $\gamma \to 0$.)
\end{remark}
\begin{proof}
    Let $\phi_0^{\delta}$ denote the $L^2(\mu_{\delta})$ orthogonal projection of $\phi_0$ onto $L^2_0(\mu_{\delta})$.
    We consider the decomposition~\eqref{eq:generator} of the generator and note that
    \begin{equation}
        \label{eq:remainder_expansion}
        \mathcal L^{\delta}(\phi_\delta - \phi_0^{\delta})
        = \mathcal L^{\delta}(\phi_\delta - \phi_0)
        = - \delta \, \mathcal L_1 \phi_0
        = \delta \, \grad \mathcal W \, \cdot \grad_{\vect p} \phi_0.
    \end{equation}
    It follows from~\eqref{eq:remainder_expansion} that
     \begin{align*}
         \norm{\phi_\delta - \phi_0^{\delta}}[L^2(\mu_{\delta})]
         &\leq \lvert \delta \rvert \norm{(\mathcal L^{\delta})^{-1}} [\mathcal B\left(L^2_0(\mu_{\delta})\right)]
         \norm{\grad \mathcal W \, \cdot \grad_{\vect p} \phi_0}[L^2(\mu_{\delta})].
     \end{align*}
     From the results in~\cite[Section 3.1]{BFLS20} which,
     as discussed in~\cref{ssub:bias_of_the_standard_estimator},
     were shown in various other works,
     it is clear that
     $\norm*{(\mathcal L^{\delta})^{-1}}[\mathcal B \left( L^2_{0}(\mu_{\delta})\right)] $ is bounded from above by $K(\delta) \max\{\gamma, \gamma^{-1}\}$ for all~$\gamma \in (0, \infty)$,
     with a constant $K(\delta)$ depending on $\delta$.
     A careful inspection of the result in~\cite{BFLS20} and its proof reveal that $K(\delta)$,
     which depends on $\delta$ through the Poincaré constant of $\mu_{\delta}$ among other things,
     is in fact such that $\sup \{K(\delta): -1 \leq \delta \leq 1\} < \infty$.
     Therefore,
     using in addition that $\grad \mathcal W$ is uniformly bounded,
     we obtain that
     \begin{align}
         \label{eq:convergence_l2}
         \forall (\gamma, \delta) \in (0, \infty) \times [-1, 1], \qquad
         \norm{\phi_\delta - \phi_0^{\delta}}[L^2(\mu_{\delta})]
         &\leq C \lvert \delta \rvert \max\{\gamma, \gamma^{-1}\}
         \norm{\grad_{\vect p} \phi_0}[L^2(\mu_{\delta})].
     \end{align}
    Here and throughout this proof, $C$ denotes a constant independent of $\gamma$ and $\delta$,
    whose value can change from occurrence to occurrence.
    In order to show~\eqref{eq:convergence_delta},
    it remains to bound $\norm*{\phi_0^{\delta}- \phi_0}$;
    the result then follows from the triangle inequality.
    To this end, note that
    \(
        \d \mu_{\delta} / \d \mu_0  = 1 + \delta f_{\delta}(q_1, q_2)
    \)
    for some appropriate smooth function~$f_{\delta}$ that is $0$ if $\delta = 0$ and is otherwise given explicitly by
    \[
        f_{\delta}(q_1, q_2) = \frac{1}{\delta} \left( \frac{Z(0)}{Z(\delta)} \e^{- \beta \delta \mathcal W(q_1, q_2)} - 1 \right),
        \qquad Z(\delta) = \int_{\torus \times \torus} \e^{- \beta\bigl(\mathcal V(q_1) + \mathcal V(q_2) + \delta \mathcal W(q_1, q_2)\bigr)} \, \d q_1 \d q_2.
    \]
    Denoting by $M$ the supremum of $\beta |\mathcal W|$ over $\torus^2$,
    we have $\e^{-\lvert \delta \rvert M} \leq Z(\delta)/Z(0) \leq \e^{\lvert \delta \rvert M}$,
    which implies that $\lvert f_{\delta}(q_1, q_2) \rvert \leq \lvert \delta \rvert^{-1}(\e^{2 \lvert \delta \rvert M} - 1) \leq 2M\e^{2M}$
    uniformly for $(\delta, q_1, q_2) \in [-1, 1] \times \torus \times \torus$.
    Therefore, using that $\phi_0$ has average 0 with respect to $\mu_0$,
    \[
        \norm*{\phi_0^\delta - \phi_0} [L^2(\mu_\delta)]
        = \abs{ \int_{\torus^2 \times \real^2} \phi_0 \, \d \mu_{\delta} }
        = \abs{ \int_{\torus^2 \times \real^2} \phi_0 \,\bigl( 1 + \delta f_{\delta}(q_1, q_2) \bigr)\d \mu_{0}}
        \leq C \lvert \delta \rvert \norm{\phi_0}[L^2(\mu_0)].
    \]
    which leads to~\eqref{eq:convergence_delta} in view of~\cref{remark:norms_equivalent}.

    We now show~\eqref{eq:convergence_gradient}.
    Taking the $L^2(\mu_{\delta})$ inner product of both sides of~\eqref{eq:remainder_expansion} with $\phi_\delta - \phi_0^{\delta}$,
    noting that $\grad_{\vect p} \phi_0^{\delta} = \grad_{\vect p} \phi_0$
    and using~\eqref{eq:convergence_l2},
    we obtain,
    for all $(\gamma, \delta) \in (0, \infty) \times [-1, 1]$,
    \begin{align*}
        \notag
        \gamma \beta^{-1} \norm*{\grad_{\vect p} \phi_\delta - \grad_{\vect p} \phi_0}[L^2(\mu_{\delta})]^2
        &\leq \lvert \delta \rvert \norm*{\grad \mathcal W \cdot \grad_{\vect p} \phi_0}[L^2(\mu_{\delta})] \norm*{\phi_\delta - \phi_0^{\delta}}[L^2(\mu_{\delta})] \\
        \label{eq:intermediate_delta}
        &\leq C \lvert \delta \rvert \norm*{\grad_{\vect p} \phi_0}[L^2(\mu_{\delta})] \norm*{\phi_\delta - \phi_0^{\delta}}[L^2(\mu_{\delta})] \\
        &\leq  C \delta^2 \max\{\gamma, \gamma^{-1}\} \norm*{\grad_{\vect p} \phi_0}[L^2(\mu_{\delta})]^2,
    \end{align*}
    which gives the claimed result.
\end{proof}

A direct corollary of~\cref{proposition:convergence_gradient} is that,
if the exact solution to the Poisson equation in one spatial dimension is employed for constructing the control variate,
i.e.\ if $\psi_{\vect e} = \phi_0$ in the control variate~\eqref{eq:definition_control_variate},
then by~\cref{proposition:asymptotic_variance} and~\eqref{eq:convergence_gradient}
we have
\begin{align*}
    \lim_{T \to \infty} \var \bigl[v(T)\bigr]
    &\leq 4 \gamma \beta^{-1} \norm{\grad_{\vect p} \phi_{\delta} - \grad_{\vect p} \phi_0}[L^2(\mu_{\delta})]^2 (d[\phi_{\delta}] + d[\phi_0]) \\
    &\leq 4 C \delta^2 \gamma^{-1} \beta^{-1} \norm{\grad_{\vect p} \phi_{0}}^2 (d[\phi_{\delta}] + d[\phi_0]) \\
    & = 4 C \delta^2 \gamma^{-2} d[\phi_{0}] (d[\phi_{\delta}] + d[\phi_0]),
\end{align*}
which shows that, in this case,
the asymptotic relative standard deviation admits an upper bound scaling as $\lvert \delta \rvert / \gamma$
in the limit as $\lvert \delta \rvert/\gamma \to 0$,
provided that $\delta \in [-1, 1]$ and $\gamma \in (0, 1)$.

\Cref{fig:time_bias_variance_2d} depicts the behavior of the mobility with respect to $\gamma$ for various values of $\delta$.
For the sake of clarity,
only the data calculated without a control variate,
i.e.\ with the simple estimator~$u(T)$ given in~\eqref{eq:simple_estimator},
is depicted in this figure.
The reason for this choice is that,
as we shall see in the next figure,
the estimator $u(T)$ has the smallest variance when $\gamma \ll 1$ and~$\delta/\gamma \gg 1$.
It appears clearly from the figure that the mobility behaves,
over the range of values of $\gamma$ we consider,
as $\gamma^{-\sigma}$ for $\sigma \in (0, 1]$ in the underdamped regime,
with an exponent~$\sigma$ that decreases as $\delta > 0$ increases (at least for sufficiently small values of~$\delta$). As mentioned in the introduction, it is an open problem to identify classes of potentials for which a universal scaling of the diffusion coefficient with respect to the friction can be rigorously shown in the nonseperable case.

The variance of the estimators obtained using the control approach described above,
where~$\psi_{\vect e}$ is constructed from an approximate solution to the Poisson equation in the one-dimensional setting,
is presented in~\cref{fig:time_bias_deviation_2d}.
The control variates corresponding to the data in the left and right panels
are constructed with the Galerkin approach of~\cref{sub:galerkin_approach}
and the underdamped approach of~\cref{sub:underdamped_approach}, respectively.
We observe that unless $\delta = 0$, in which case we recover the one-dimensional case,
there seems to be, for every~$\delta > 0$ and for each of the two choices of $\psi_{\vect e}$,
a threshold value of $\gamma$ below which the control variate ceases to be useful.
This is not unexpected, in view of~\cref{proposition:convergence_gradient}.
Although the control variates we consider are advantageous in certain regimes,
further work is necessary in order to develop efficient estimators in the small~$\gamma$, constant $\delta$ regime.

\begin{figure}[ht]
    \centering
    \includegraphics[width=0.99\linewidth]{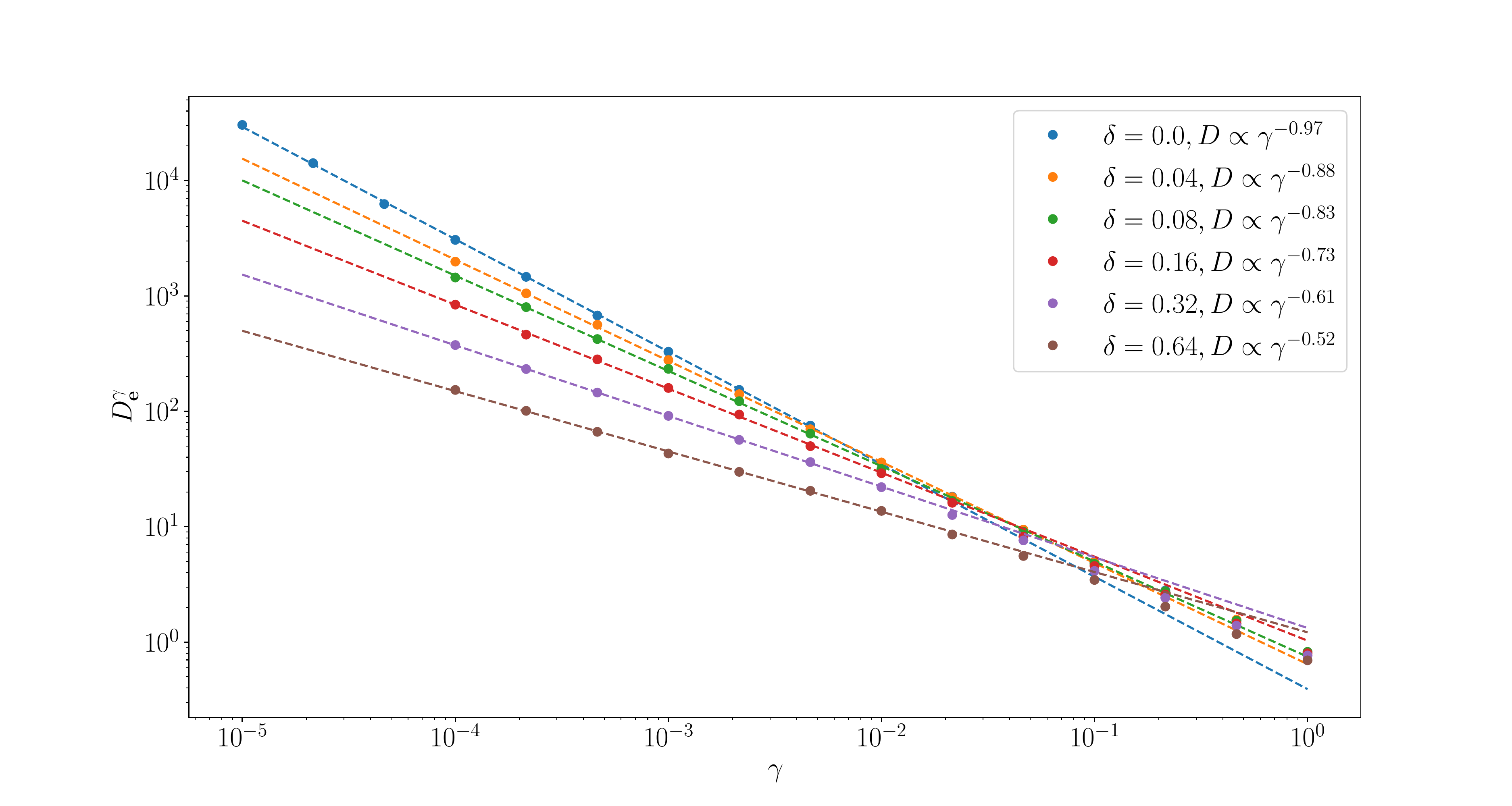}
    \caption{
        Effective diffusion coefficient as a function of $\gamma$,
        for various values of $\delta$.
        The straight dashed lines, which correspond to functions of the form $\gamma \mapsto C \gamma^{-\sigma}$ with powers $\sigma$ indicated in the legend,
        are obtained by linear interpolation (in the log-log plot) using only values of $\gamma$ less than or equal to $10^{-2}$.
    }%
    \label{fig:time_bias_variance_2d}
\end{figure}

\begin{figure}[ht]
    \centering
    \includegraphics[width=0.49\linewidth]{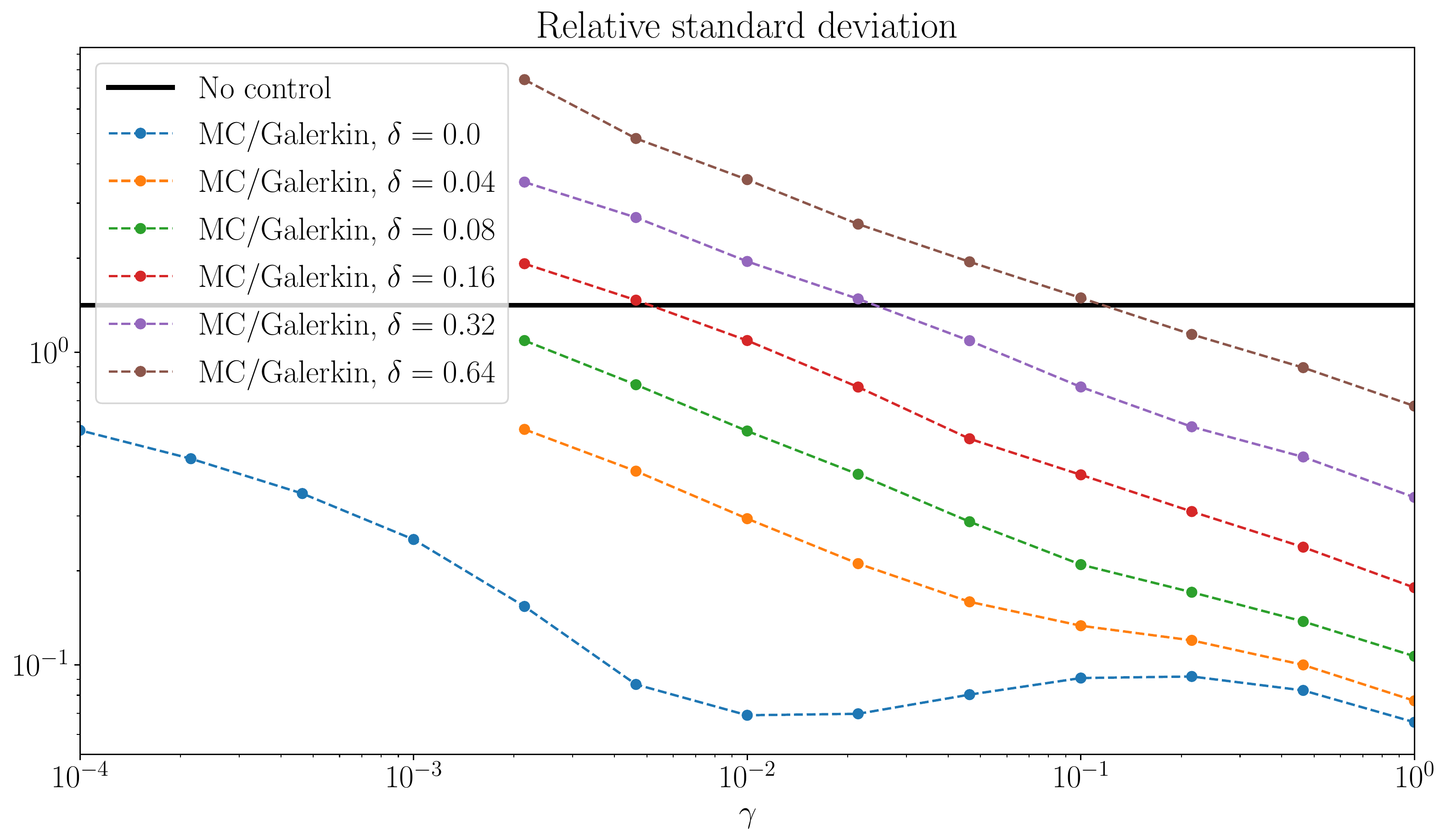}
    \includegraphics[width=0.49\linewidth]{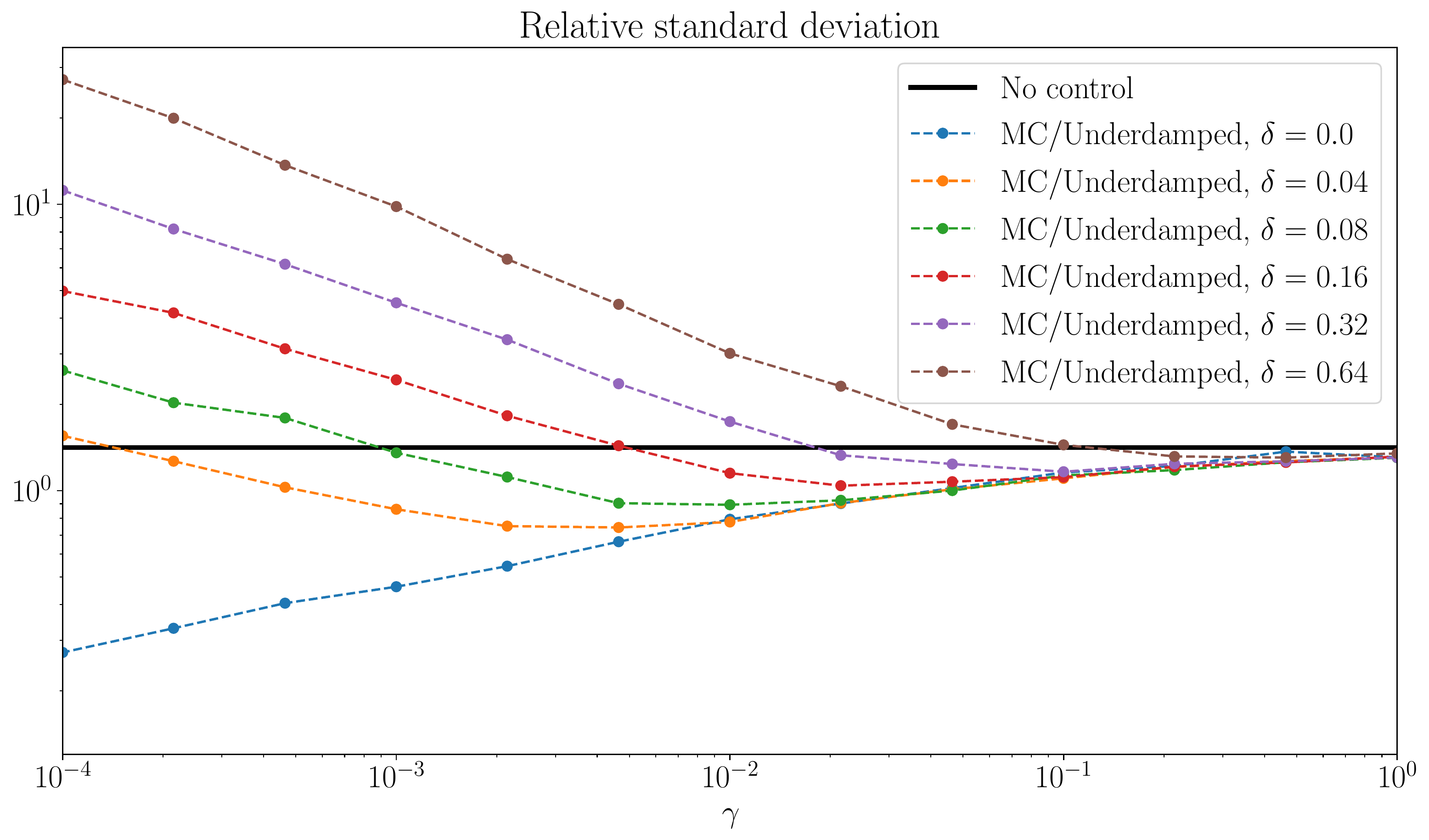}
    \caption{
        Relative standard deviation of the estimator $v(T)$ in~\eqref{eq:improved_estimator} for two-dimensional Langevin dynamics,
        when the approximate solution to the Poisson equation is constructed by tensorization from the solution of the one-dimensional equation.
        The one-dimensional solution is approximated using either the Galerkin approach of \cref{sub:galerkin_approach} (\textbf{left} panel)
        or the underdamped approach of \cref{sub:underdamped_approach} (\textbf{right} panel).
    }%
    \label{fig:time_bias_deviation_2d}
\end{figure}

\section{Conclusions and perspectives for future work}%
\label{sec:conclusions_and_perspectives_for_future_work}
In this work,
we showed how techniques based on control variates can be employed for improving estimators of dynamical properties,
here the mobility of Langevin dynamics based on Einstein's formula.
The control variate approach we propose requires the knowledge of an approximate solution to a Poisson equation involving the generator of the dynamics.
We studied several practical approaches for constructing this approximate solution,
and we obtain general bounds on the bias and variance of the proposed estimator in terms of the approximation error.

In the one-dimensional setting,
we demonstrated the efficiency of control variates
(i) obtained by a Fourier/Hermite spectral method for the Poisson equation, and
(ii) based on an explicit expression for the limiting solution to the Poisson equation in the underdamped limit.
For both Langevin and generalized Langevin dynamics,
the latter approach leads to a significant variance reduction in the very small friction regime $\gamma \leq 10^{-3}$,
in which fully deterministic Galerkin methods are inaccurate.

The numerical experiments we presented for the one-dimensional generalized Langevin dynamics also corroborate prior findings in~\cite{GPGSUV21},
obtained through formal asymptotics,
concerning the asymptotic behavior of the mobility in the small friction limit.
More precisely, they indicate that the mobility scales as $1/\gamma$ as $\gamma \to 0$ when the parameter~$\nu$ encoding memory is fixed,
with a prefactor different from that corresponding to Langevin dynamics.

The two-dimensional setting for Langevin dynamics is much more challenging because of the high dimensionality of the state space of the dynamics
and the lack of theoretical results for the underdamped limit in the case of a non-separable potential.
Nevertheless, the control variates developed for one-dimensional Langevin dynamics may still be applied with appropriate tensorization,
and we show by means of numerical experiments that
they lead to estimators with reduced variance provided that
the parameter multiplying the non-separable part of the potential is small with respect to the friction $\gamma$.

We anticipate that, in the future,
approaches based on physics-informed neural networks (PINN)~\cite{MR3759415,MR3881695} will be useful for constructing more accurate solutions to the Poisson equation~\eqref{eq:poisson_equation} in two spatial dimensions than is possible using a Galerkin method,
providing scope for greater variance reduction for Monte Carlo estimators of the mobility.
The body of literature on the use of neural networks in the context of high-dimensional PDEs has grown very rapidly in recent years,
and there is increasing evidence that these approaches are able to overcome the curse of dimensionality in a variety of PDE applications~\cite{MR3847747,MR4338293,2018arXiv180907321J,2019arXiv190810828G,2019arXiv190110854H,2021arXiv210614473D,MR4203091,pmlr-v145-zhai22a}.

\paragraph{Acknowledgements.}
The work of GS and UV was partially funded by the European Research Council (ERC) under the European Union's Horizon 2020 research and innovation programme (grant agreement No 810367),
and by the Agence Nationale de la Recherche under grant ANR-21-CE40-0006 (SINEQ).
The work of GAP was partially supported by JPMorgan Chase \& Co under J.P. Morgan A.I. Research Awards in 2019 and 2021 and by the EPSRC, grant number EP/P031587/1.
The work of UV was partially funded by the Fondation Sciences Mathématiques de Paris (FSMP),
through a postdoctoral fellowship in the ``mathematical interactions'' programme.

\appendix
\section{Proof of \texorpdfstring{\cref{proposition:semigroup_meanzero_observable}}{Proposition 2.1}}%
\label{sec:auxiliary_technical_results}

The proof is based on several lemmata.
Before presenting these, we introduce some notation and recall useful background material.
For a measure $\pi$, we define the weighted Sobolev space~$H^i(\pi)$ as the subspace of $L^2(\pi)$
of functions whose derivatives up to order $i$ are in $L^2(\pi)$.
The associated norm is given by
\[
    \norm{u}_i^2 = \sum_{|\vect \alpha| \leq i} \norm*{\partial^{\vect \alpha} f}^2,
\]
where we use the standard multi-index notation.
Throughout this section, $\norm{\dummy}$ and $\ip{\dummy}{\dummy}$ are the norm and inner product of~$L^2(\pi)$.
The probability measure $\pi$ implicit in this notation will be either obvious in the context or explicitly specified.
We recall that $\pi$ is said to satisfy a Poincaré inequality with constant $R$ if
\begin{equation}
    \label{eq:poincare}
    \tag{P$_R$}
    \forall f \in H^1(\pi) \cap L^2_0(\pi), \qquad
    \norm*{f}^2 \leq \frac{1}{R} \norm{\grad f}^2.
\end{equation}
It is well known that the Gaussian measure $\kappa$ in~\eqref{eq:definition_prob_measures} satisfies~\eqref{eq:poincare} with a constant $R_{\kappa} = \beta$;
see, for example,~\cite[Lemma 2.1]{MR4071827} for a short proof.
We will use the following standard result which follows,
for example, from~\cite[Chapter 9]{lorenzi2006analytical}.
\begin{lemma}
    \label{lemma:semigroup}
    It holds that
    \[
        D(\mathcal L_{\rm FD}) := \Bigl\{ f \in L^2_0(\kappa): \mathcal L_{\rm FD} f \in L^2_0(\kappa) \Bigr\} = H^2(\kappa) \cap L^2_0(\kappa).
    \]
    In addition, the unbounded operator
    \(
        \mathcal L_{\rm FD}
    \)
    with domain $D(\mathcal L_{\rm FD})$ has a discrete spectrum and
    generates a contraction semigroup~$(\e^{t \mathcal L_{\rm FD}})_{t \geq 0}$ on $L^2_0(\kappa)$,
    with $\lVert \e^{t \mathcal L_{\rm FD}} f \rVert \leq \e^{-t} \lVert f \rVert$ for all $f \in L^2_0(\kappa)$ and $t \geq 0$.
\end{lemma}

The eigenfunctions of $\mathcal L_{\rm FD}$ are given by rescaled Hermite polynomials,
which form a complete orthonormal set of~$L^2(\kappa)$.
The eigenvalue associated with the constant polynomial is~0,
and all the other eigenvalues are integer numbers less than or equal to $-1$.
We denote the normalized eigenfunctions by $(\mathscr H_{\vect \alpha})_{\vect \alpha \in \nat^d}$
and the corresponding eigenvalues by $(-\lambda_{\vect \alpha})_{\vect \alpha \in \nat^d}$.
Then
\begin{equation}
    \label{eq:series_expansion}
    \forall f \in L^2_0(\kappa), \qquad
    \e^{t \mathcal L_{\rm FD}} f
    = \sum_{\abs{\vect \alpha} \geq 1} \e^{-\lambda_{\vect \alpha} t} \ip{f}{\mathscr H_{\vect \alpha}} \mathscr H_{\vect \alpha}.
\end{equation}
This formula can be employed to show the following estimates.
\begin{corollary}
    If $f \in L^2_0(\kappa)$, then $\e^{t \mathcal L_{\rm FD}} f \in H^1(\kappa)$ for $t > 0$ and
    \begin{subequations}
    \begin{equation}
        \label{lemma:elliptic_reg}
        \forall t > 0, \qquad
        \norm{\grad_{\vect p} \e^{t \mathcal L_{\rm FD}} f}_{L^2(\kappa)} \leq \frac{C}{\sqrt{t}} \norm{f}_{L^2(\kappa)}.
    \end{equation}
    On the other hand,
    if $f \in H^1(\kappa)$,
    then
    \begin{equation}
        \label{lemma:overdamped_langevin_decay_derivatives}
        \forall t \geq 0, \qquad
        \norm{\grad_{\vect p} \e^{t \mathcal L_{\rm FD}} f}_{L^2(\kappa)} \leq \e^{-t} \norm{\grad_{\vect p} f}_{L^2(\kappa)}.
    \end{equation}
    \end{subequations}
\end{corollary}
\begin{proof}
    Assume first that $f \in D(\mathcal L_{\rm FD})$.
    Then $\e^{t \mathcal L_{\rm FD} } f \in D(\mathcal L_{\rm FD})$ for all $t > 0$ by the general properties of operator semigroups~\cite{MR710486} and,
    using~\eqref{eq:series_expansion}, we obtain
    \begin{equation}
        \label{eq:gradient_series}
        \norm{\grad_{\vect p} \e^{t \mathcal L_{\rm FD}} f}_{L^2(\kappa)}^2 = \beta \ip{\mathcal L_{\rm FD} \e^{t \mathcal L_{\rm FD}} f}{\e^{t \mathcal L_{\rm FD}} f}
        = \sum_{\abs{\vect \alpha} \geq 1} \beta \lambda_{\vect \alpha} \e^{- 2\lambda_{\vect \alpha} t} \lvert \ip{f}{\mathscr H_{\vect \alpha}} \rvert^2.
    \end{equation}
    The first claim follows from the inequality
    \(
        \lambda \e^{- 2\lambda t} \leq \frac{\e^{-1}}{2 t},
    \)
    which is valid for all $\lambda \in \real$.
    The second claim follows from the fact that $\e^{-2\lambda_{\vect \alpha} t} \leq \e^{-2t}$ for all $|\vect \alpha| \geq 1$.

    Consider now the case where $f \in L^2(\kappa)$ does not necessarily belong to $D(\mathcal L_{\rm FD})$
    and let
    \[
        f_n = \sum_{0 \leq \abs{\vect \alpha} \leq n} \ip{f}{\mathscr H_{\vect \alpha}} \mathscr H_{\vect \alpha} \in D(\mathcal L_{\rm FD})
    \]
    For fixed $t > 0$, the sequence $(\e^{t \mathcal L_{\rm FD}}f_n)_{n\in \nat}$ is a Cauchy sequence in $H^1(\kappa)$ by~\eqref{lemma:elliptic_reg},
    and so it converges to a limit in $H^1(\kappa)$ which necessarily coincides with the $L^2(\kappa)$ limit $\e^{t \mathcal L_{\rm FD}} f$.
    We conclude by applying~\eqref{lemma:elliptic_reg} and~\eqref{lemma:overdamped_langevin_decay_derivatives} to $f_n$ and taking the limit $n \to \infty$.
\end{proof}



Let us now introduce additional notation.
We define $\nabla_{\vect q}^* = \beta \nabla V(\vect q) - \nabla_{\vect q}$ and $\nabla_{\vect p}^* = \beta \vect p - \nabla_{\vect p}$,
and note that these operators are formally the $L^2(\mu)$ adjoints of $\nabla_{\vect q}$ and $\nabla_{\vect p}$.
Similarly, in one dimension we write $\partial_q^* = \beta V'(q) - \partial_q$ and $\partial_p^* = \beta p - \partial_p$.

The next lemma provides an intermediate result for proving~\cref{proposition:semigroup_meanzero_observable},
and concerns Langevin dynamics over the state space $\torus^d \times \real^d$.
This result is sharper than what would be obtained from a simple application of~\eqref{eq:decay_semigroup_general},
but not yet sufficient for obtaining optimal estimates for the bias of estimator~\eqref{eq:simple_estimator}.
To lighten notations,
we confine ourselves from now on to the one-dimensional setting in the proofs,
but these carry over \emph{mutatis mutandis} to the multi-dimensional case.
\begin{lemma}
    \label{lemma:initial_lemma}
    Assume that $f \in L^2(\mu)$ is a smooth function such that $\grad_{\vect q} f \in L^2(\mu)$ and
    \begin{equation}
        \label{eq:assumption_f}
        \forall \vect q \in \torus^d, \qquad
        \Pi_{\vect p} f(\vect q) = \int_{\real^d} f(\vect q, \vect p) \, \kappa(\d \vect p) = 0.
    \end{equation}
    Then there exist constants $C > 0$ and $\lambda > 0$ independent of $\gamma$ and $f$ such that
    \begin{align}
        \label{eq:better_decay_estimate}
        \forall \gamma \geq 1, \qquad
        \forall t \geq 0, \qquad
        \norm*{\e^{t\mathcal L} f}
        \leq C \bigl( \norm{f} + \norm{\grad_{\vect q} f} \bigr)
        \left( \e^{- \gamma t} + \gamma^{-1} \e^{-\frac{\lambda t}{\gamma}} \right).
    \end{align}
\end{lemma}
The key feature of~\eqref{eq:better_decay_estimate} is that,
when $\gamma \gg 1$,
the prefactor $\gamma^{-1}$ multiplying the slow exponential $\e^{-\frac{\lambda t}{\gamma}}$ is small,
showing that $\e^{t\mathcal L} f$ is small after a time of order $\mathcal O(\gamma^{-1})$.
\begin{proof}
    We prove the result for functions $f(q, p)$ of the form
    \begin{equation}
        \label{eq:expansion}
        f(q, p) = \e^{\frac{\beta}{2}H(q,p)}\sum_{i=0}^{N} \sum_{j=1}^{N} c_{ij} \, g_i(q) h_j(p),
    \end{equation}
    where $g_i$ and $h_j$ are the basis functions defined in~\eqref{eq:definition_trigonometric_functions} and~\eqref{eq:definition_hermite_functions}.
    The space of functions of this form is dense in $\bigl\{u \in (\id - \Pi_p) L^2(\mu): \partial_q u \in L^2(\mu) \bigr\}$
    endowed with the norm $\norm{u}_{1,q} := \norm{u} + \norm{\partial_q u}$,
    so the general result follows by density.

    We expect that $\e^{t \mathcal L} f \approx \e^{t \gamma \mathcal L_{\rm FD}} f$ in an appropriate sense for $\gamma \gg 1$.
    Therefore, let us introduce $v(t) = \e^{t \mathcal L} f(q, p) - \e^{t \gamma \mathcal L_{\rm FD}} f(q, p)$
    and show that~$v(t)$ is small.
    In the expression $\e^{t \gamma \mathcal L_{\rm FD}} f(q, p)$,
    the variable $q$ should be viewed as a parameter.
    The function~$v$ satisfies the initial value problem
    \[
        \partial_t v = \mathcal L v +  \mathcal L_{\rm Ham} \bigl(\e^{t \gamma \mathcal L_{\rm FD}} f\bigr), \qquad v(0) = 0.
    \]
    By Duhamel's formula,
    \[
        v(t) = \int_{0}^{t} \e^{(t-s) \mathcal L}  \Bigl( \mathcal L_{\rm Ham} \bigl(\e^{s \gamma \mathcal L_{\rm FD}} f\bigr) \Bigr) \, \d s,
    \]
    and therefore
    \begin{equation}
        \label{eq:intermediate_decay_correlation}
        \e^{t \mathcal L} f =  \e^{t \gamma \mathcal L_{\rm FD}} f
        + \int_{0}^{t} \e^{(t-s) \mathcal L}  \Bigl( \mathcal L_{\rm Ham} \bigl(\e^{s \gamma \mathcal L_{\rm FD}} f\bigr) \Bigr) \, \d s.
    \end{equation}
    The first term is bounded as
    \begin{align}
        \notag
        \norm{\e^{t \gamma \mathcal L_{\rm FD}} f}^2
        &= \int_{\torus} \int_{\real}  \abs{\e^{t \gamma \mathcal L_{\rm FD}} f(q, p) }^2 \kappa(\d p) \,\nu(\d q)
        = \int_{\torus} \norm{\e^{t \gamma \mathcal L_{\rm FD}} f(q, \dummy) }[L^2(\kappa)]^2 \, \nu(\d q) \\
        \label{eq:bound_first_term}
        &\leq \int_{\torus} \e^{-2 \gamma t} \norm{f(q, \dummy) }[L^2(\kappa)]^2 \, \nu(\d q)
        = \e^{-2\gamma t} \norm{f}^2,
    \end{align}
    where we employed \cref{lemma:semigroup} in the second line.
    We now bound the second term on the right-hand side of~\eqref{eq:intermediate_decay_correlation}.
    The commutator relation~$\commut{\partial_p}{\partial_p^*} = \beta$ implies that
    \[
        \forall \varphi \in L^2(\kappa), \qquad
        \norm{\partial_p^* \varphi}_{L^2(\kappa)}^2 = \norm{\partial_p \varphi}_{L^2(\kappa)}^2 + \beta \norm{\varphi}_{L^2(\kappa)}^2.
    \]
    Using this equation together with the definition of $\partial_q^*$,
    we have
    \begin{align*}
        \norm{\partial_q \partial_p^* \e^{t \gamma \mathcal L_{\rm FD}} f}
        &\leq \norm{\partial_q \partial_p \e^{t \gamma \mathcal L_{\rm FD}} f} + \sqrt{\beta} \norm{\partial_q \e^{t \gamma \mathcal L_{\rm FD}} f}, \\
        \norm{\partial_q^* \partial_p \e^{t \gamma \mathcal L_{\rm FD}} f}
        &\leq \norm{\partial_q \partial_p \e^{t \gamma \mathcal L_{\rm FD}} f}
        + \beta \norm{V'}[\infty] \norm{\partial_p \e^{t \gamma \mathcal L_{\rm FD}} f}.
    \end{align*}
    From these equations, we deduce
    \begin{align}
        \notag
        \norm{\mathcal L_{\rm Ham} \e^{t \gamma \mathcal L_{\rm FD}} f}
        &= \beta^{-1} \lVert (\partial_q \partial_p^* - \partial_q^* \partial_p) \e^{t \gamma \mathcal L_{\rm FD}} f \rVert \\
        \label{eq:reasoning_action_lham}
        &\leq 2 \beta^{-1} \norm{\partial_q \partial_p \e^{t \gamma \mathcal L_{\rm FD}} f}
        + \beta^{-1/2} \norm{\partial_q \e^{t \gamma \mathcal L_{\rm FD}} f}
        + \norm*{V'}[\infty] \norm{\partial_p \e^{t \gamma \mathcal L_{\rm FD}} f}.
    \end{align}
    Since $f$ is assumed to be a finite linear combination of the form~\eqref{eq:expansion},
    we can freely change the order of the operators $\partial_q$ and $\e^{t \gamma \mathcal L_{\rm FD}}$;
    in particular, the function $\e^{t \gamma \mathcal L_{\rm FD}} f$ is differentiable in $q$.
    Letting $t_* = \min\{ \gamma^{-1}, t\}$, we have
    \begin{align*}
        \norm{\partial_q \partial_p \e^{t \gamma \mathcal L_{\rm FD}} f}
        &= \norm{\partial_p \e^{t \gamma \mathcal L_{\rm FD}} \partial_q f}
        = \norm{\partial_p \e^{(t - t_*) \gamma \mathcal L_{\rm FD}} (\e^{t_* \gamma \mathcal L_{\rm FD}} \partial_q f)} \\
        &\leq \e^{- \gamma (t - t_*)} \norm{\partial_p \e^{t_* \gamma \mathcal L_{\rm FD}} \partial_q f}
        \leq C \frac{\e^{- \gamma (t - t_*)}}{\sqrt{\gamma t_*}} \norm{\partial_q f},
    \end{align*}
    where we applied~\eqref{lemma:overdamped_langevin_decay_derivatives} in the first inequality,
    then \eqref{lemma:elliptic_reg} in the second inequality.
    Bounding the other terms in~\eqref{eq:reasoning_action_lham} using the same method
    and observing that~$\e^{\gamma t_*} \leq \e$,
    we obtain
    \begin{align}
        \label{eq:bound_action_lham}
        \norm{\mathcal L_{\rm Ham} \e^{t \gamma \mathcal L_{\rm FD}} f}
        &\leq C  \frac{\e^{-\gamma t}}{\sqrt{\min\{1, \gamma t\}}}  \bigl( \norm{f} + \norm{\partial_q f} \bigr).
    \end{align}
    Going back to~\eqref{eq:intermediate_decay_correlation} and using~\eqref{eq:decay_semigroup_general},
    we have, for $\gamma \geq 1$,
    \begin{align}
        \label{eq:bound_intermediate}
        \norm*{ \e^{t \mathcal L} f}
        &\leq  \e^{-\gamma t} \norm{f}
        + C  \bigl( \norm{f} + \norm{\partial_q f}\bigr) \int_{0}^{t} \e^{-\frac{\ell}{\gamma}(t-s)}  \, \left(\frac{\e^{-\gamma s}}{\sqrt{\min\{1, \gamma s\}}}\right) \, \d s.
    \end{align}
    It is clear that~\eqref{eq:better_decay_estimate} holds for all $t \in [0, \gamma^{-1}]$,
    provided that the prefactor on the right-hand side of that equation is sufficiently large.
    To obtain~\eqref{eq:better_decay_estimate} for times larger than $1/\gamma$,
    we bound the integral on the right-hand side of~\eqref{eq:bound_intermediate} by decomposing the interval~$[0, t]$ as $[0, \frac{1}{\gamma}] \cup [\frac{1}{\gamma}, t]$:
    \begin{align}
        \notag
        &\int_{0}^{t} \e^{-\frac{\ell}{\gamma}(t-s)}  \, \left( \frac{\e^{-\gamma s}}{\sqrt{\min\{1, \gamma s\}}} \right) \, \d s
        \leq
        \e^{-\frac{\ell}{\gamma} \left(t-\frac{1}{\gamma}\right)}
         \int_{0}^{\frac{1}{\gamma}}  \frac{1}{\sqrt{\gamma s\, }} \, \d s
         + \int_{\frac{1}{\gamma}}^{t} \e^{-\frac{\ell}{\gamma}(t-s) - \gamma s} \d s \\
         \notag
        &\qquad \leq
        \frac{2}{\gamma}\e^{-\frac{\ell}{\gamma} \left(t-\frac{1}{\gamma}\right)}
         + \int_{0}^{t} \e^{-\frac{\lambda}{\gamma}(t-s) - \gamma s} \d s \\
         \label{eq:bound_integral}
        &\qquad \leq
        \frac{2}{\gamma}\e^{-\frac{\ell}{\gamma} \left(t-\frac{1}{\gamma}\right)}
         + \int_{0}^{\infty} \e^{-\frac{\lambda}{\gamma}(t-s) - \gamma s} \d s
        \leq C  \gamma^{-1} \e^{-\frac{\lambda t}{\gamma}},
    \end{align}
    where $\lambda = \min \{\frac{1}{2}, \ell \}$.
    Defining $\lambda$ in this manner ensures that the integral on the last line is bounded for all $\gamma \in [1, \infty)$.
    The conclusion then follows from~\eqref{eq:bound_intermediate} and~\eqref{eq:bound_integral}.
\end{proof}

We are now ready to prove~\cref{proposition:semigroup_meanzero_observable}.
\begin{proof}
    [Proof of \cref{proposition:semigroup_meanzero_observable}]
    We again show the result for functions $f$ and $h$ of the form~\eqref{eq:expansion},
    noting that the general result follows by density of functions of this type.
    Recall that, by assumption, both $f$ and $h$ are in $\Pi_{p}^\perp L^2(\mu)$,
    and their derivatives with respect to $q$ are in $L^2(\mu)$.
    From~\eqref{eq:intermediate_decay_correlation}, we obtain
    \begin{align}
        \notag
        \ip{\e^{t \mathcal L} f}{h}
        &= \ip{\e^{t \gamma \mathcal L_{\rm FD}} f}{h}
        + \int_{0}^t \ip{\e^{(t-s) \mathcal L} \Bigl( \mathcal L_{\rm Ham} \bigl(\e^{s \gamma \mathcal L_{\rm FD}} f\bigr) \Bigr)}{h} \, \d s \\
        &= \ip{\e^{t \gamma \mathcal L_{\rm FD}} f}{h}
        + \int_{0}^t \ip{ \mathcal L_{\rm Ham} \bigl(\e^{s \gamma \mathcal L_{\rm FD}} f\bigr) }{\e^{(t-s) \mathcal L^*}  h} \, \d s.
    \end{align}
    The first term is bounded as in~\eqref{eq:bound_first_term}.
    In order to bound the second term,
    we denote the $L^2(\mu)$ adjoint of the generator $\mathcal L$ by~
    \(
        \mathcal L^* = - \mathcal L_{\rm Ham} + \gamma \mathcal L_{\rm FD}
    \),
    and employ the fact that \cref{lemma:initial_lemma} is valid also with $\mathcal L^*$ substituted for $\mathcal L$,
    and so
    \[
        \forall \gamma \geq 1, \quad
        \forall t \geq 0, \qquad
        \norm{\e^{(t-s)\mathcal L^*} h}
        \leq C \bigl( \norm{h} + \norm{\partial_q h} \bigr)
        \left( \e^{- \gamma (t-s)} + \gamma^{-1} \e^{-\frac{\lambda (t-s)}{\gamma}} \right).
    \]
    Combined with~\eqref{eq:bound_action_lham},
    this inequality gives
    \begin{align*}
        & \abs{ \ip{\e^{t \mathcal L} f}{h}}
        \leq \norm{f} \norm{h} \e^{-\gamma t} \\
        & \qquad + C \bigl( \norm{f} + \norm{\partial_q f} \bigr) \bigl( \norm{h} + \norm{\partial_q h} \bigr)
            \int_{0}^{t}  \left( \frac{\e^{-\gamma s}}{\sqrt{\min\{1, \gamma s\}}} \right)
         \left(  \e^{- \gamma (t-s)} + \gamma^{-1} \e^{-\frac{\lambda (t-s)}{\gamma}} \right) \, \d s.
    \end{align*}
    The first term in the integral is bounded as
    \[
        \int_{0}^{t}  \left( \frac{\e^{-\gamma s}}{\sqrt{\min\{1, \gamma s\}}} \right)
          \e^{- \gamma (t-s)}  \, \d s
        = \e^{- \gamma t} \left( \int_{0}^{\gamma^{-1}} \frac{1}{\sqrt{\gamma s}} \, \d s + t - \gamma^{-1} \right)
        = \e^{-\gamma t} \left( t + \frac{1}{\gamma} \right).
    \]
    For any $\alpha \in (0, 1)$, the right-hand side of this equation may be bounded from above by $C_{\alpha} \e^{- \alpha \gamma t}$ .
    The second term in the integral can be bounded as in~\eqref{eq:bound_integral} in the proof of \cref{lemma:initial_lemma},
    which concludes the proof of~\cref{proposition:semigroup_meanzero_observable}.
\end{proof}

\section{Proof of~\texorpdfstring{the bound \eqref{eq:refined_bound}}{the technical bound on the bias}}%
\label{sec:proof_technical_result}
In order to prove the bound,
we write
\begin{align*}
    \ip{\e^{t \mathcal L}(\vect e^\t \vect p)}{\vect e^\t \vect p)} - \ip{\e^{t \mathcal L} \mathcal L \psi_{\vect e}}{\mathcal L \psi_{\vect e}}
    = &\ip{\e^{t \mathcal L} (\vect e^\t \vect p)}{\vect e^\t \vect p + (\id - \Pi_{\vect p}) \mathcal L \psi_{\vect e}} \\
    & - \ip{\e^{t \mathcal L} \bigl(\vect e^\t \vect p + (\id - \Pi_{\vect p}) \mathcal L \psi_{\vect e}\bigr)}{(\id - \Pi_{\vect p}) \mathcal L \psi_{\vect e}} \\
    & - \ip{\e^{t \mathcal L} (\id - \Pi_{\vect p}) \mathcal L \psi_{\vect e}}{\Pi_{\vect p} \mathcal L \psi_{\vect e}}
    - \ip{\e^{t \mathcal L} \Pi_{\vect p} \mathcal L \psi_{\vect e}}{\mathcal L \psi_{\vect e}}.
\end{align*}
The last two terms are bounded using the general bound on the Langevin semigroup~\eqref{eq:decay_semigroup_general}:
\begin{align*}
    &\max \left\{
    \abs{\ip{\e^{t \mathcal L} (\id - \Pi_{\vect p}) \mathcal L \psi_{\vect e}}{\Pi_{\vect p} \mathcal L \psi_{\vect e}}},
    \abs{\ip{\e^{t \mathcal L} \Pi_{\vect p} \mathcal L \psi_{\vect e}}{\mathcal L \psi_{\vect e}}}
    \right\} \\
    &\qquad \qquad \qquad \qquad \leq L \exp \left( - \ell \min\{\gamma, \gamma^{-1}\} t\right) \norm{\Pi_{\vect p} \mathcal L \psi_{\vect e}} \norm{\mathcal L \psi_{\vect e}}.
\end{align*}
The first two terms are bounded using \cref{proposition:semigroup_meanzero_observable}:
\begin{align*}
    \abs{\ip{\e^{t \mathcal L} (\vect e^\t \vect p)}{\vect e^\t \vect p + (\id - \Pi_{\vect p}) \mathcal L \psi_{\vect e}}}
    &\leq C \beta^{-1/2} \norm*{\vect e^\t \vect p + \mathcal L \psi_{\vect e}}[1,\vect q] \, \zeta(t), \\
    \abs{\ip{\e^{t \mathcal L} \bigl(\vect e^\t \vect p + (\id - \Pi_{\vect p}) \mathcal L \psi_{\vect e}\bigr)}{(\id - \Pi_{\vect p}) \mathcal L \psi_{\vect e}}}
    &\leq C \norm*{\vect e^\t \vect p + \mathcal L \psi_{\vect e}}[1,\vect q] \, \norm*{\mathcal L \psi_{\vect e}}[1,\vect q] \, \zeta(t),
\end{align*}
where $\zeta(t) = \gamma^{-2} \e^{- a \gamma^{-1} t } + \e^{- a \gamma t}$.
Here we employed the fact that, for any $f \in \lp{2}{\mu}$ such that also $\grad_{\vect q} f \in \lp{2}{\mu}$,
it holds
\[
    \norm{\grad_{\vect q} (\id - \Pi_{\vect p}) f}
    = \norm{(\id - \Pi_{\vect p}) \grad_{\vect q} f}
    \leq \norm{\grad_{\vect q} f}.
\]
and so $\norm{(\id - \Pi_{\vect p}) f}[1,\vect q] \leq \norm{f}[1,\vect q]$.
The bound~\eqref{eq:refined_bound} is then obtained by repeating the reasoning in the proof of~\cref{lemma:bias_improved}.

\section{Proof of \texorpdfstring{equation~\eqref{eq:limit_variance}}{the limit of the variance}}%
\label{sec:proof_of_equation_variance}
By definition of the variance,
we have
\[
    \var \Bigl[ \left\lvert X^T_1 \right\rvert^2 - \left\lvert X^T_2 \right\rvert^2 \Bigr]
    = \expect \Bigl[ \left\lvert X^T_1 \right\rvert^4 + \left\lvert X^T_2 \right\rvert^4 - 2 \left\lvert X^T_1 X^T_2 \right\rvert^2 \Bigr]
    - \Bigl( \expect \Bigl[ \left\lvert X^T_1 \right\rvert^2 - \left\lvert X^T_2 \right\rvert^2 \Bigr] \Bigr)^2,
\]
and so it is sufficient to prove that,
in the limit as $T \to \infty$ and for $i \in \{1, 2\}$,
\begin{align}
    \label{eq:target_equation}
    \expect \left[ \left\lvert X^T_i \right\rvert^4 \right]
    \to \expect \left[ \left\lvert X^{\infty}_i \right\rvert^4 \right],
    \quad
    \expect \left[ \left\lvert X^T_i \right\rvert^2 \right]
    \to \expect \left[ \left\lvert X^{\infty}_i \right\rvert^2 \right],
    \quad
    \expect \left[ \left\lvert X^T_1 X^T_2 \right\rvert^2 \right]
    \to \expect \left[ \left\lvert X^T_1 X^T_2 \right\rvert^2 \right].
\end{align}
By the continuous mapping theorem,
it holds from the convergence in law of $X^{T}$ to $X^{\infty}$ that
\[
     \lvert X^T_i \rvert^4
    \xrightarrow[T \to \infty]{\rm Law}
    \left\lvert X^{\infty}_i \right\rvert^4 ,
    \qquad
     \lvert X^T_i \rvert^2
    \xrightarrow[T \to \infty]{\rm Law}
    \left\lvert X^{\infty}_i \right\rvert^2 ,
    \qquad
     \left\lvert X^T_1 X^T_2 \right\rvert^2
    \xrightarrow[T \to \infty]{\rm Law}
    \left\lvert X^T_1 X^T_2 \right\rvert^2 .
\]
In order to prove~\eqref{eq:target_equation},
it is therefore sufficient~\cite[Theorem 3.5]{MR1700749}
to show that the random variables $\lvert X^T_i \rvert^4$, $\lvert X^T_i \rvert^2$ and $\lvert X^T_1 X^T_2 \rvert^2$
are uniformly integrable over $T \in [1, \infty)$.
We check this condition carefully for $\lvert X^T_2 \rvert^4$,
noting that the same reasoning can be applied to the other terms.
A sufficient condition for uniform integrability is that there is $\varepsilon > 0$ such that
\[
    \sup_{T \in [1, \infty)} \expect \left[ \left\lvert X^T_2 \right\rvert^{4+\varepsilon} \right] < \infty.
\]
By definition,
it holds that
\[
    X^T_2 =
    \frac{\psi_{\vect e}(\vect q_0, \vect p_0) - \psi_{\vect e}(\vect q_T, \vect p_T)}{\sqrt{2T}}
    + \sqrt{\frac{\gamma}{\beta T}} \int_{0}^{T} \grad_{\vect p} \psi_{\vect e} (\vect q_t, \vect p_t) \cdot \d \vect w_t,
\]
and so, using manipulations similar to those in the proof of~\cref{proposition:variance},
including a moment inequality for It\^o integrals~\cite[Theorem 7.1]{MR2380366} and the standing assumption of a stationary initial condition,
we have, for $T \geq 0$,
\[
    \expect \left[ \lvert X^T_2 \rvert^{4+ \varepsilon} \right]
    \leq
    C \left( \norm{\psi_{\vect e}}[L^{4+\varepsilon}(\mu)]^{4 + \varepsilon} +
        \norm{\grad_{\vect p} \psi_{\vect e}}[L^{4+\varepsilon}(\mu)]^{4+\varepsilon}
        \right) < \infty.
\]



\printbibliography
\end{document}